\documentclass[conference]{IEEEtran}
\let\labelindent\relax
\usepackage{times}
\usepackage{color}
\usepackage{balance}
\usepackage{url}
\usepackage{enumitem}
\usepackage{siunitx}
\usepackage{caption}

\usepackage{epsfig,tabularx,amssymb,amsmath,subfigure,multirow,algorithm,algorithmic,graphicx}

\begin{document}

\title{Task Assignment on Multi-Skill Oriented Spatial Crowdsourcing (Technical Report)}

\author{\IEEEauthorblockN{Peng Cheng\IEEEauthorrefmark{1},
Xiang Lian\IEEEauthorrefmark{2},
Lei Chen\IEEEauthorrefmark{1}, 
Jinsong Han\IEEEauthorrefmark{3} and
Jizhong Zhao\IEEEauthorrefmark{3}}
\IEEEauthorblockA{\IEEEauthorrefmark{1}
Hong Kong University of Science and Technology, Hong Kong, China\\ Email: \{pchengaa, leichen\}@cse.ust.hk}
\IEEEauthorblockA{\IEEEauthorrefmark{2}University of Texas Rio Grande Valley, Texas, USA\\
Email: xiang.lian@utrgv.edu}
\IEEEauthorblockA{\IEEEauthorrefmark{3}Xi'an Jiaotong University, Shaanxi, China\\
Email: \{hanjinsong, zjz\}@mail.xjtu.edu.cn}}

\newcommand{\nop}[1]{}
\renewcommand{\algorithmicrequire}{\textbf{Input:}}
\renewcommand{\algorithmicensure}{\textbf{Output:}}
\newtheorem{definition}{Definition}
\newtheorem{theorem}{Theorem}
\newtheorem{lemma}{Lemma}
\newtheorem{remark}{Remark}
\newtheorem{corollary}{Corollary}
\newcommand{\narrowParagraph}{0}
\setlength{\parskip}{0ex}

\maketitle

\begin{abstract}
	With the rapid development of mobile devices and crowdsourcing platforms, the spatial crowdsourcing has attracted much attention from the database community. Specifically, the spatial crowdsourcing refers to sending location-based requests to workers, based on their current positions. In this paper, we consider a spatial crowdsourcing scenario, in which each worker has a set of qualified skills, whereas each spatial task (e.g., repairing a house, decorating a room, and performing entertainment shows for a ceremony) is time-constrained, under the budget constraint, and required a set of  skills. Under this scenario, we will study an important problem, namely \textit{multi-skill spatial crowdsourcing} (MS-SC), which finds an optimal worker-and-task assignment strategy, such that skills between workers and tasks match with each other, and workers' benefits are maximized under the budget constraint. We prove that the MS-SC problem is NP-hard and intractable. Therefore, we propose three effective heuristic approaches, including greedy, $g$-divide-and-conquer and cost-model-based adaptive algorithms to get worker-and-task assignments. Through extensive experiments, we demonstrate the efficiency and effectiveness of our MS-SC processing approaches on both real and synthetic data sets.
	
\end{abstract}

\section{Introduction}

With the popularity of GPS-equipped smart devices and wireless
mobile networks \cite{deng2013maximizing,kazemi2012geocrowd}, nowadays people can easily identify and participate
in some location-based tasks that are close to their current
positions, such as taking photos/videos, repairing houses, and/or
preparing for parties at some spatial locations. Recently, a new
framework, namely \textit{spatial crowdsourcing}
\cite{kazemi2012geocrowd}, for employing workers to conduct spatial
tasks, has emerged in both academia (e.g., the database community
\cite{chen2014gmission}) and industry (e.g., TaskRabbit
\cite{taskrabbit}). A typical spatial crowdsourcing platform (e.g.,
gMission \cite{chen2014gmission} and MediaQ \cite{kim2014mediaq})
assigns a number of moving \textit{workers} to do \textit{spatial
	tasks} nearby, which requires workers to physically move to some
specified locations and accomplish these tasks.

Note that, not all spatial tasks are as simple as taking a photo or
video clip (e.g., street view of Google Maps
\cite{GoogleMapStreetView}), monitoring traffic conditions (e.g.,
Waze \cite{waze}), or reporting local hot spots (e.g., Foursquare
\cite{foursquare}), which can be easily completed by providing
answers via camera, sensing devices in smart phones, or naked eyes,
respectively. In contrast, some spatial tasks can be rather complex,
such as repairing a house, preparing for a party, and performing
entertainment shows for a ceremony, which may consist of several
steps/phases/aspects, and require demanding professional skills from
workers. In other words, these complex tasks cannot be simply
accomplished by normal workers, but require the skilled workers with
specific expertise (e.g., fixing roofs or setting up the stage).

\begin{figure}[t]\centering
	\scalebox{0.8}[0.8]{\includegraphics{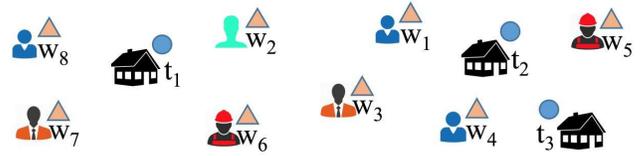}}\vspace{-1ex}
	\caption{\small An Example of Repairing a House in the Multi-Skill Spatial Crowdsourcing System.}\vspace{-5ex}
	\label{fig:bbq_example}
\end{figure}

\begin{table}\vspace{-3ex}
	\parbox[b]{.44\linewidth}{
		\caption{\small Worker/Task Skills}\label{tab:marker_skill}\vspace{-2ex}
		{\small\scriptsize
			\begin{tabular}{l|l}
				{\bf worker/task} & {\bf  skill key set}\\
				\hline \hline
				$w_1$, $w_4$ $w_8$& \{$a_1$, $a_4$, $a_6$\}\\
				$w_2$& \{$a_5$\}\\
				$w_3$, $w_7$& \{$a_2$,  $a_3$\}\\
				$w_5$, $w_6$& \{$a_1$,  $a_5$\}\\
				\hline
				$t_1$, $t_2$, $t_3$& \{$a_1$ $\sim$ $a_6$\}\\
				\hline
			\end{tabular}
		}
		\vspace{-1ex}
	}
	\hfill
	\parbox[b]{.52\linewidth}{
		\caption{\small Descriptions of Skills}\label{tab:skill_description}\vspace{-2ex}
		{\small\scriptsize
			\begin{tabular}{c|l}
				{\bf skill key} & {\bf \quad  skill description}\\
				\hline \hline
				$a_1$& painting walls\\
				$a_2$& repairing roofs\\
				$a_3$& repairing floors\\
				$a_4$& installing pipe systems\\
				$a_5$& installing electronic components\\
				$a_6$& cleaning\\
				\hline
			\end{tabular}
		}
		\vspace{-1ex}
	}
	\vspace{-4ex}
\end{table}

Inspired by the phenomenon of complex spatial tasks, in this paper,
we will consider an important problem in the
spatial crowdsourcing system, namely \textit{multi-skill spatial
	crowdsourcing} (MS-SC), which assigns multi-skilled workers to those
complex tasks, with the matching skill sets and high scores of the
worker-and-task assignments.

In the sequel, we will illustrate  the
MS-SC problem by a motivation example of repairing a house.

\nop{
	
	The spatial tasks in the spatial crowdsourcing system can be
	checking the display shelves at neighborhood stores, taking a
	photo/video (e.g., street view of Google Maps
	\cite{GoogleMapStreetView}), monitoring traffic conditions (e.g.,
	Waze \cite{waze}), or reporting local hot spots (e.g., Foursquare
	\cite{foursquare}). Since simple tasks are easy for workers to
	accomplish (i.e., workers only need to move to some locations, and
	provide the answer via naked eyes, camera, or sensing devices in
	smart phones), they do not require expert skills from workers. Thus,
	any worker can be assigned to do simple tasks.

	On the other hand, however, spatial tasks can be as complex as
	repairing a house, decorating a room, performing entertainment shows
	for a ceremony, or reporting an ad-hoc news, which consist of
	several steps/phases/aspects, and require demanding professional
	skills from workers. In this case, the spatial crowdsourcing
	platform has to assign multi-skilled workers only to those complex
	tasks that exactly require their skills, which is more challenging
	than the assignment with simple tasks.

}

\vspace{0.5ex}\noindent {\bf Example 1 (Repairing a House).} {\it
	Consider a scenario of the spatial crowdsourcing in Figure
	\ref{fig:bbq_example}, where a user wants to repair a house he/she
	just bought, in order to have a good living environment for his/her
	family. However, it is not an easy task to repair the house, which
	requires many challenging works (skills), such as repairing
	roofs/floors, replacing/installing pipe systems and
	electronic components, painting walls, and finally cleaning rooms.
	There are many skilled workers that can accomplish one or some of
	these skill types. In this case, the user can post a spatial task
	$t_1$, as shown in Figure \ref{fig:bbq_example}, in the spatial
	crowdsourcing system, which specifies a set of required skills
	(given in Tables \ref{tab:marker_skill} and \ref{tab:skill_description}) for the
	house-repairing task, a valid time period to repair, and the maximum
	budget that he/she would like to pay.
	
	In Figure \ref{fig:bbq_example}, around the spatial location of
	task $t_1$, there are 8 workers, $w_1 \sim w_8$, each of whom has a
	different set of skills as given in Table \ref{tab:marker_skill}.
	For example, worker $w_1$ has the skill set $\{\text{painting walls},$
	$\text{installing pipe systems},$ $\text{cleaning}\}$.

	To accomplish the spatial task $t_1$ (i.e., repair the house), the
	spatial crowdsourcing platform needs to select a best subset of
	workers $w_i$ ($1\leq i\leq 8$), such that the union of their skill
	sets can cover the required skill set of task $t_1$, and, moreover,
	workers can travel to the location of $t_1$ with the maximum net
	payment under the constraints of arrival times, workers' moving ranges,  and budgets. For
	example, we can assign task $t_1$ with 3 workers $w_2$, $w_7$, and
	$w_8$, who are close to $t_1$, and whose skills can cover all the
	required skills of $t_1$.
	
}

\nop{
	
	\vspace{0.5ex}\noindent {\bf Example 2 (Performing Ad Hoc News
		Reporting).} {\it A task requester wants to report a breaking news
		suddenly occurring at a remote location. The news reporting needs
		different skills such as taking photos, recording audio/videos,
		interviewing people, news background analysis, and writing the news.
		In this situation, the task requester can post a spatial task of
		news reporting with the specified location to the spatial
		crowdsourcing system, which includes the information such as the
		budget, a required skill set, and the valid period. The system can
		then assign those workers, who are near the specified location and
		have the required news reporting skills, to the news reporting task.}
	
}

Motivated by the example above, in this paper, we will
formalize the MS-SC problem, which aims to efficiently assign
workers to complex spatial tasks, under the task constraints of
valid time periods and maximum budgets, such that the required skill
sets of tasks are fully covered by those assigned workers, and the
total score of the assignment (defined as the total profit of
workers) is maximized.

Note that, existing works on spatial crowdsourcing focused on
assigning workers to tasks to maximize the total number of completed
tasks \cite{kazemi2012geocrowd}, the number of performed tasks for a
worker with an optimal schedule \cite{deng2013maximizing}, or the
reliability-and-diversity score of assignments
\cite{cheng2014reliable}. However, they did not take into account
multi-skill covering of complex spatial tasks, time/distance constraints, and the assignment
score with respect to task budgets and workers' salaries (excluding
the traveling cost). Thus, we cannot directly apply prior solutions to
solve our MS-SC problem.

In this paper, we first prove that our MS-SC problem in the spatial
crowdsourcing system is NP-hard, by reducing it from the \textit{Set
	Cover Problem} (SCP) \cite{karp1972reducibility}. As a result, the
MS-SC problem is not tractable, and thus very challenging to achieve
the optimal solution. Therefore, in this paper, we will tackle the
MS-SC problem by proposing three effective approximation approaches,
greedy, $g$-\textit{divide-and-conquer} ($g$-D\&C), and
cost-model-based adaptive algorithms, which can efficiently compute
worker-and-task assignment pairs with the constraints/goals of
skills, time, distance, and budgets.

Specifically, we make the following contributions.
\begin{itemize}[leftmargin=*]
	
	\item We formally define the \textit{multi-skill spatial
		crowdsourcing} (MS-SC) problem in Section \ref{sec:problem_def}, under the constraints of multi-skill covering, time, distance, and budget for spatial workers/tasks in the spatial crowdsourcing system.
	
	\item We prove that the MS-SC problem is NP-hard, and thus
	intractable in Section \ref{sec:reduction}.
	
	\item We propose efficient approximation approaches, namely greedy,
	$g$-divide-and-conquer, and cost-model-based adaptive
	algorithms to tackle the MS-SC problem in Sections \ref{sec:greedy}, \ref{sec:g_D&C}, and \ref{sec:adpative}, respectively.
	
	\item We conduct extensive experiments on real and synthetic data sets,
	and show the efficiency and effectiveness of our MS-SC approaches in Section \ref{sec:exper}.
\end{itemize}

Section \ref{sec:framework} introduces a general framework for our MS-SC problem in spatial crowdsourcing systems. Section \ref{sec:related} reviews previous works on spatial
crowdsourcing. Finally, Section \ref{sec:conclusion} concludes this
paper.

\section{Problem Definition}
\label{sec:problem_def}

In this section, we present the formal definition of the multi-skill
spatial crowdsourcing, in which we assign multi-skilled
workers with time-constrained complex spatial tasks.

\subsection{Multi-Skilled Workers}

We first define the multi-skilled workers in spatial crowdsourcing
applications. Assume that $\Psi=\{a_1, a_2, ..., a_k\}$ is a
universe of $k$ abilities/skills. Each worker has one or multiple
skills in $\Psi$, and can provide services for spatial tasks that
require some skills in $\Psi$.

\begin{definition} $($Multi-Skilled Workers$)$ Let
	$W_p$ $=\{w_1,$ $w_2,$ $...,$ $w_n\}$ be a set of $n$ multi-skilled
	workers at timestamp $p$. Each worker $w_i$ ($1\leq i\leq n$) has a
	set, $X_i$ $(\subseteq \Psi)$, of skills, is located at position
	$l_i(p)$ at timestamp $p$, can move with velocity $v_i$, and has a maximum moving distance $d_i$.
	\qquad $\blacksquare$ \label{definition:worker}
\end{definition}

In Definition \ref{definition:worker}, the multi-skilled workers
$w_i$ can move dynamically with speed $v_i$ in any direction, and at
each timestamp $p$, they are located at spatial places $l_i(p)$, and
prefer to move at most $d_i$ distance from $l_i(p)$. They can freely
join or leave the spatial crowdsourcing system. Moreover, each
worker $w_i$ is associated with a set, $X_i$, of skills, such as
taking photos, cooking, and decorating rooms.

\subsection{Time-Constrained Complex Spatial Tasks}

Next, we define complex spatial tasks in the spatial crowdsourcing
system, which are constrained by deadlines of arriving at task
locations and budgets.

\begin{definition}
	$($Time-Constrained Complex Spatial Tasks$)$ Let $T_p=\{t_1, t_2,
	..., t_m\}$ be a set of time-constrained complex spatial tasks at
	timestamp $p$. Each task $t_j$ ($1\leq j\leq m$) is located
	at a specific location $l_j$, and workers are expected to reach the
	location of task $t_j$ before the arrival deadline $e_j$. Moreover, to complete the
	task $t_j$, a set, $Y_j$ $(\subseteq \Psi)$, of skills is
	required for those assigned workers. Furthermore, each task $t_j$ is
	associated with a budget, $B_j$, of salaries for workers. \qquad
	$\blacksquare$ \label{definition:task}
\end{definition}

As given in Definition \ref{definition:task}, usually, a task
requester creates a time-constrained spatial task $t_j$, which
requires workers physically moving to a specific location $l_j$, and
arriving at $l_j$ before the arrival deadline $e_j$. Meanwhile, the task
requester also specifies a budget, $B_j$, of salaries, that is, the
maximum allowance that he/she is willing to pay for workers. This
budget, $B_j$, can be either the reward cash or bonus points in the
spatial crowdsourcing system.

Moreover, the spatial task $t_j$ is often complex, in the sense that
it might require several distinct skills (in $Y_j$) to be conducted.
For example, a spatial task of repairing a house may require
several skills, such as repairing floors, painting walls and cleaning.

\subsection{The Multi-Skill Spatial Crowdsourcing Problem}

In this subsection, we will formally define the
multi-skill spatial crowdsourcing (MS-SC) problem, which assigns spatial
tasks to workers such that workers can cover the skills required by
tasks and the assignment strategy can achieve high scores.

\vspace{0.5ex}\noindent {\bf Task Assignment Instance Set.} Before
we present the MS-SC problem, we first introduce the concept of task
assignment instance set.

\begin{definition}
	$($Task Assignment Instance Set$)$ At timestamp $p$, given a worker
	set $W_p$ and a task set $T_p$, a \textit{task assignment instance
		set}, denoted by $I_p$, is a set of worker-and-task assignment pairs
	in the form $\langle w_i, t_j\rangle$, where each worker $w_i \in
	W_p$ is assigned to at most one spatial task $t_j\in T_p$.
	
	Moreover, we denote $CT_p$ as the set of completed tasks $t_j$ that can be reached
	before the arrival deadlines $e_j$, and accomplished by those assigned workers in $I_p$.\qquad
	$\blacksquare$ \label{definition:instance}
\end{definition}

Intuitively, the task assignment instance set $I_p$ is one possible
(valid) worker-and-task assignment between worker set $W_p$ and task
set $T_p$. Each pair $\langle w_i, t_j\rangle$ is in $I_p$, if and
only if this assignment satisfies the constraints of task $t_j$,
with respect to distance (i.e., $d_i$), time (i.e., $e_j$), budget
(i.e., $B_j$), and skills (i.e., $Y_j$).

In particular, for each pair $\langle w_i, t_j\rangle$, worker $w_i$
must arrive at location $l_j$ of the assigned task $t_j$ before its
arrival deadline $e_j$, and can support the skills required by task $t_j$,
that is, $X_i \bigcap Y_j \neq \emptyset$. The distance between
$l_i(p)$ and $l_j$ should be less than $d_i$. Moreover, for all
pairs in $I_p$ that contain task $t_j$, the required skills of task
$t_j$ should be fully covered by skills of its assigned workers,
that is, $Y_j \subseteq \cup_{\forall \langle w_i, t_j\rangle \in
	I_p}X_i$.

To assign a worker $w_i$ to a task $t_j$, we need to pay him/her
salary, $c_{ij}$, which is related to the traveling cost from the
location, $l_i(p)$, of worker $w_i$ to that, $l_j$, of task $t_j$.
The traveling cost, $c_{ij}$, for vehicles can be calculated by the
unit gas price per gallon times the number of gallons needed for the
traveling. For the public transportation, the cost $c_{ij}$ can be
computed by the fees per mile times the traveling distance. For
walking, we can also provide the compensation fee for the worker
with the cost $c_{ij}$ proportional to his/her traveling distance.

Without loss of generality, we assume that the cost, $c_{ij}$, is
proportional to the traveling distance, $dist(l_i(p), l_j)$, between
$l_i(p)$ and $l_j$, where $dist(x, y)$ is a distance function
between locations $x$ and $y$. Formally, we have: $c_{ij} = C_i\cdot
dist(l_i(p), l_j)$, where $C_i$ is a constant (e.g.,
gas/transportation fee per mile).

Note that, for simplicity, in this paper, we use Euclidean distance
as our distance function (i.e., $dist(x, y)$). We can easily extend
our proposed approaches in this paper by considering other distance
function (e.g., road-network distance), under the framework of the
spatial crowdsourcing system, and would like to leave the topics of
using other distance metics as our future work.

\vspace{0.5ex}\noindent {\bf The MS-SC Problem.} In the sequel, we
give the definition of our \textit{multi-skill spatial
	crowdsourcing} (MS-SC) problem.

\begin{definition}
	(Multi-Skill Spatial Crowdsourcing Problem) Given a time interval $P$, the problem of
	\textit{multi-skill spatial crowdsourcing} (MS-SC) is to assign the
	available workers $w_i\in W_p$ to spatial tasks $t_j\in T_p$, and to obtain a
	task assignment instance set, $I_p$, at each timestamp $p \in P$, such that:
	
	\begin{enumerate}
		\item any worker $w_i\in W_p$ is assigned to only one spatial
		task $t_j\in T_p$ such that his/her arrival time at location $l_j$
		before the arrival deadline $e_j$, the moving distance is
		less than the worker's maximum moving distance $d_i$, and 
		all workers assigned to $t_j$ have skill sets fully covering $Y_j$;
		\item the total traveling cost of all the assigned workers to task $t_j$ does not exceed the budget of
		the task, that is, $\sum_{\forall \langle w_i, t_j\rangle \in
			I_p}c_{ij}$ $\leq B_j$; and
		\item the total score, $\sum_{p \in P}S_p$, of the task assignment
		instance sets $I_p$ within the time interval $P$ is maximized,
	\end{enumerate}
	
	where it holds that:
	
	{\small\vspace{-4ex}
		\begin{eqnarray}
			S_p &=& \sum_{t_j \in CT_p}B'_j, and \label{eq:assignment_score}\\
			B'_j &=& B_j - \sum_{\langle w_i, t_j\rangle \in I_p}c_{ij}.\label{eq:flexible_budget}
		\end{eqnarray}
		\vspace{-2ex}
		\label{definition:MS_SC}}
\end{definition}

In Definition \ref{definition:MS_SC}, our MS-SC problem aims to
assign workers $w_i$ to tasks $t_j$ such that: (1) workers $w_i$ are
able to reach locations, $l_j$, of tasks $t_j$ on time and
cover the required skill set $Y_j$, and the moving distance is less than $d_i$; (2) the total traveling cost
of all the assigned workers should not exceed budget $B_j$; and
(3) the total score, $\sum_{p \in P}S_p$, of the task-and-worker assignment within time interval $P$ should be maximized.

After the server-side assignment at a timestamp $p$, those assigned
workers would change their status to unavailable, and move to the
locations of spatial tasks. Next, these workers will become
available again, only if they finish/reject the assigned tasks.

\vspace{0.5ex}\noindent {\bf Discussions on the Score $S_p$.}
Eq.~(\ref{eq:assignment_score}) calculates the score, $S_p$, of a
task-and-worker assignment by summing up \textit{flexible budgets},
$B_j'$ (given by Eq.~(\ref{eq:flexible_budget})), of all the
completed tasks $t_j\in CT_p$, where the \textit{flexible budget} of
task $t_j$ is the remaining budget of task $t_j$ after paying
workers' traveling costs. Maximizing scores means maximizing the 
number of accomplished tasks while minimizing the traveling cost of workers.

Intuitively, each task $t_j$ has a maximum budget $B_j$, which
consists of two parts, the traveling cost of the assigned workers
and the flexible budget. The former cost is related to the
total traveling distance of workers, whereas the latter one can be
freely and flexibly used for rewarding workers for their
contributions to the task. Here, the distribution of the flexible
budget among workers can follow existing incentive mechanisms in
crowdsourcing \cite{rula2014no, yang2012crowdsourcing}, which
stimulate workers who did the task better (i.e., with more
rewards).

Note that, in Eq.~(\ref{eq:assignment_score}), the score, $S_p$, of
the task assignment instance set $I_p$ only takes into account those
tasks that can be completed by the assigned workers (i.e., tasks in
set $CT_p$). Here, a task can be completed, if the assigned workers
can reach the task location before the deadline and finish the task
with the required skills.

Since the spatial crowdsourcing system is quite dynamic, new
tasks/workers may arrive at next timestamps. Thus, if we cannot find
enough/proper workers to do the task at the current timestamp $p$,
the task is still expected to be successfully assigned with workers
and completed in future timestamps. Meanwhile, the task requester
can be also informed by the spatial crowdsourcing system to increase
the budget (i.e., with higher budget $B_j$, we can find more skilled
candidate workers that satisfy the budget constraint). Therefore, in
our definition of score $S_p$, we would only consider those tasks in
$CT_p$ that can be completed by the assigned workers at timestamp
$p$, and maximize this score $S_p$.

\subsection{Hardness of Multi-Skill Spatial Crowdsourcing Problem}
\label{sec:reduction}

With $m$ time-constrained complex spatial tasks and $n$
multi-skilled workers, in the worst case, there are an exponential
number of possible task-and-worker assignment strategies, which
leads to high time complexity, $O((m + 1)^n)$. In this subsection,
we prove that our MS-SC problem is NP-hard, by reducing a well-known
NP-hard problem, \textit{set cover problem} (SCP)
\cite{vazirani2013approximation}, to the MS-SC problem.

\begin{lemma} (Hardness of the MS-SC Problem)
	The problem of the Multi-Skill Spatial Crowdsourcing (MS-SC) is
	NP-hard. \label{lemma:np}\vspace{-1ex}
\end{lemma}
\begin{proof}
	Please refer to Appendix A.
\end{proof}

Since the MS-SC problem involves multiple spatial tasks whose skill
sets should be covered, we thus cannot directly use existing
approximation algorithms for SCP (or its variants) to solve the
MS-SC problem. What is more, we also need to find an assignment
strategy such that workers and tasks match with each other (in terms
of traveling time/cost, and budge constraints), which is more
challenging.

Thus, due to the NP-hardness of our MS-SC problem,  in subsequent
sections, we will present a general framework for MS-SC processing and design 3 heuristic algorithms, namely greedy,
$k$-divide-and-conquer, and cost-model-based adaptive approaches, to
efficiently retrieve MS-SC answers.

\begin{table}
	\begin{center}\vspace{-5ex}
		\caption{\small Symbols and Descriptions.} \label{table0}
		{\small\scriptsize
			\begin{tabular}{l|l}
				{\bf Symbol} & {\bf \qquad \qquad \qquad\qquad\qquad Description} \\ \hline \hline
				$T_p$   & a set of $m$ time-constrained spatial tasks $t_j$ at timestamp $p$\\
				$W_p$   & a set of $n$ dynamically moving workers $w_i$ at timestamp $p$\\
				$e_j$   & the deadline of arriving at the location of task $t_j$\\
				$l_i(p)$  & the position of worker $w_i$ at timestamp $p$\\
				$l_j$   & the position of task $t_j$\\
				$X_i$   & a set of skills that worker $w_i$ has\\
				$Y_j$   & a set of the required skills for task $t_j$ \\
				$d_i$   & the maximum moving distance of worker $w_i$ \\
				$B_j$   & the maximum budget of task $t_j$ \\
				$I_p$   & the task assignment instance set at timestamp $p$ \\
				$CT_p$   & a set of tasks that are assigned with workers at timestamp $p$ and\\
				& \qquad can be completed by these assigned workers\\
				$C_i$ & the unit price of the traveling cost of worker $w_i$\\
				$c_{ij}$ & the traveling cost from the location of worker $w_i$ to that of task $t_j$\\
				$S_p$   & the score of the task assignment instance set $I_p$\\
				$\Delta S_p$   & the score increase when changing the pair assignment\\ \hline
				\hline
			\end{tabular}
		}
	\end{center}\vspace{-6ex}
\end{table}

Table \ref{table0} summarizes the commonly used symbols.

\section{Framework of Solving MS-SC Problems}
\label{sec:framework}

In this section, we present a general framework,  
namely {\sf MS-SC\_Framework}, in Figure \ref{alg:framework} for solving the
MS-SC problem, which greedily assigns workers with spatial tasks for
multiple rounds. For each round, at timestamp $p$, we first retrieve
a set, $T_p$, of all the available spatial tasks, and a set, $W_p$,
of available workers (lines 2-3). Here, the available task set $T_p$
contains existing spatial tasks that have not been assigned with
workers in the last round, and the ones that newly arrive at the
system after the last round. Moreover, set $W_p$ includes those
workers who have accomplished (or rejected) the previously assigned
tasks, and thus are available to receive new tasks in the current
round.

In our spatial crowdsourcing system, we organize both sets $T_p$ and
$W_p$ in a cost-model-based grid index. For the sake of space
limitations, details about the index construction can be found in
Appendix E. Due to
dynamic changes of sets $T_p$ and $W_p$, we also update the grid
index accordingly (line 4). Next, we utilize the grid index to
efficiently retrieve a set, $S$, of valid worker-and-task candidate
pairs (line 5). That is, we obtain those pairs of workers and tasks,
$\langle w_i, t_j\rangle$, such that workers $w_i$ can reach the
locations of tasks $t_j$ and satisfy the constraints of skill
matching, time, and budgets for tasks $t_j$. With valid pairs in set
$S$, we can apply our proposed algorithms, that is, \textit{greedy},
\textit{$g$-divide-and-conquer}, or \textit{adaptive
	cost-model-based} approach, over set $S$, and obtain a good
worker-and-task assignment strategy in an assignment instance set
$I_p$, which is a subset of $S$ (line 6).

Finally, for each pair $\langle w_i, t_j \rangle$ in the selected
worker-and-task assignment set $I_p$, we will notify worker $w_i$ to
do task $t_j$ (lines 7-8).

\begin{figure}[ht]
	\begin{center}\vspace{-2ex}
		\begin{tabular}{l}
			\parbox{3.1in}{
				\begin{scriptsize}
					\begin{tabbing}
						12\=12\=12\=12\=12\=12\=12\=12\=12\=12\=12\=\kill
						{\bf Procedure {\sf MS-SC\_Framework}} \{ \\
						\> {\bf Input:} a time interval $P$\\
						\> {\bf Output:} a worker-and-task assignment strategy within the time interval $P$\\
						\> (1) \> \> for each timestamp $p$ in $P$\\
						\> (2) \> \> \> retrieve all the available spatial tasks to $T_p$\\
						\> (3) \> \> \> retrieve all the available workers to $W_p$\\
						\> (4) \> \> \> update the grid index for current $T_p$ and $W_p$\\
						\> (5) \> \> \> obtain a set, $S$, of valid worker-and-task pairs from the index\\
						\> (6) \> \> \> use our \textit{greedy}, \textit{$g$-divide-and-conquer} or \textit{adaptive cost-model-based} approach\\
						\> \> \> \>  \> to obtain a good assignment instance set, $I_p \subseteq S$\\
						\> (7) \> \> \> for each pair $\langle w_i, t_j \rangle$ in $I_p$\\
						\> (8) \> \> \> \> inform worker $w_i$ to conduct task $t_j$\\
						\}
					\end{tabbing}
				\end{scriptsize}
			}
		\end{tabular}
	\end{center}\vspace{-3ex}
	\caption{\small Framework for Solving the MS-SC Problem.}
	\label{alg:framework}\vspace{-5ex}
\end{figure}

\nop{
	
	The available tasks include the ones arriving between the last and
	the current rounds, and those that are unfinished after last round.
	For any task $t_j$ of the latter type, we need to maintain the set
	of workers, $W_j$, who have accepted the assigned task $t_j$, and
	update the information of the task $t_j$. When we update the
	information of unfinished task $t_j$ after last round, we reduce its
	budget $B_j$ to $B'_j=B_j - \sum_{w_i \in W_j} c_{ij}$ and change
	its required skill set $Y_j$ to $Y'_j = Y_j / \cup_{w_i \in
		W_j}X_i$.

	Moreover, we allow workers to accomplish multiple tasks within a
	period of time. However,
	each worker can only be available to accept more tasks after the
	current assigned task has been rejected or finished.
	When we assign more than one task to
	a worker in a round, if he/she uses too
	much time to finish one task,
	he/she may miss the deadlines of the other tasks,
	which could harm the overall performance of the platform.
	
	In order to facilitate the processing of the MS-SC problem,
	we present an efficient cost-model-based indexing
	mechanism, which can maintain workers and tasks and help the
	retrieval of MS-SC answers. Different from the grid-index in
	our previous work \cite{cheng2014reliable}, the grid-index in
	this work utilize bitmap synopses to time- and space- efficiently
	organize/manipulate the sets of skills for workers and tasks.
	For the details of our grid-index, please refer to
	Appendix E of our technical report \cite{arxivReport}.
	
}

\section{The Greedy Approach}
\label{sec:greedy}

In this section, we will propose a greedy algorithm, which greedily
selects one worker-and-task assignment, $\langle w_i, t_j \rangle$,
at a time that can maximize the increase of the assignment score
(i.e., $\sum_{\forall p\in P} S_p$ as given in Definition
\ref{definition:MS_SC}). This greedy algorithm can be applied in
line 6 of the framework, {\sf MS-SC\_Framework}, in Fig.
\ref{alg:framework}.

\subsection{The Score Increase}
\label{subsec:score_increase}

Before we present the greedy algorithm, we first define the
increase, $\Delta S_p$, of score $S_p$ (given in
Eq.~(\ref{eq:assignment_score})), in the case where we assign a
newly available worker $w_i$ to task $t_j$. Specifically, from
Eqs.~(\ref{eq:assignment_score}) and (\ref{eq:flexible_budget}), we
define the score increase after assigning worker $w_i$ to task $t_j$
as follows:\vspace{-2ex}

{\small
	\begin{eqnarray}
		\Delta S_p = S_p-S_{p-1} = \Delta B_j' = \frac{|X_i \cap (Y_j -
			\widetilde{Y_j})|}{|Y_j|} \cdot B_j -
		c_{ij},\label{eq:score_increase}
	\end{eqnarray}\vspace{-2ex}}

\noindent where $\widetilde{Y_j}$ is the set of skills that have
been covered by those assigned workers (excluding the new worker
$w_i$) for task $t_j$.

In Eq.~(\ref{eq:score_increase}), $\frac{|X_i \cap (Y_j -
	\widetilde{Y_j})|}{|Y_j|}$ is the ratio of skills for task $t_j$
that have not been covered by (existing) assigned workers, but can
be covered by the new worker $w_i$. Intuitively, the first term in
Eq.~(\ref{eq:score_increase}) is the pre-allocated maximum budget
based on the number of covered skills by the new worker $w_i$,
whereas the second term, $c_{ij}$, is the traveling cost from
location of $w_i$ to that of $t_j$. Thus, the score increase,
$\Delta S_p$, in Eq.~(\ref{eq:score_increase}) is to measure the
change of score (i.e., flexible budget) $S_p$, due to the assignment
of worker $w_i$ to task $t_j$.

\subsection{Pruning Strategies}
\label{sub:pruning}

The score increase can be used as a measure to evaluate and decide
which worker-and-task assignment pair should be added to the task
assignment instance set $I_p$. That is, each time our greedy
algorithm aims to choose one worker-and-task assignment pair in $S$
with the highest score increase, which will be added to $I_p$ (i.e.,
line 6 of {\sf MS-SC\_Framework} in Fig. \ref{alg:framework}).
However, it is not efficient to enumerate all valid worker-and-task
assignment pairs in $S$, and compute score increases. That is, in
the worst case, the time complexity is as high as $O(m\cdot n)$,
where $m$ is the number of tasks and $n$ is the number of workers.
Therefore, in this subsection, we present three effective pruning
methods (two for pruning workers and one for pruning tasks) to quickly filter out false alarms of worker-and-task pairs
in set $S$.

\noindent\textbf{The Worker-Pruning Strategy.} When assigning available workers
to spatial tasks, we can rule out those valid worker-and-task pairs
in $S$, which contain either \textit{dominated} or \textit{high-wage
	workers}, as given in Lemmas \ref{lemma:dominate_worker} and
\ref{lemma:expensive_worker}, respectively, below.

We say that a worker $w_a$ is \textit{dominated by} a worker $w_b$
w.r.t. task $t_j$ (denoted as $w_a \succ_{t_j} w_b$), if it holds
that $X_a \subseteq X_b$ and $c_{aj}\geq c_{bj}$, where $X_a$ and
$X_b$ are skill sets of workers $w_a$ and $w_b$, and $c_{aj}$ and
$c_{bj}$ are the traveling costs from locations of workers $w_a$ and
$w_b$ to task $t_j$, respectively.

\begin{lemma} (Pruning Dominated Workers) Given two worker-and-task
	pairs $\langle w_a, t_j\rangle$ and $\langle w_b, t_j\rangle$ in
	valid pair set $S$, if it holds that $w_a \succ_{t_j} w_b$, then we
	can safely prune the worker-and-task pair $\langle w_a,
	t_j\rangle$. \label{lemma:dominate_worker}
\end{lemma}

\begin{proof}
	Please refer to Appendix B.
\end{proof}

Lemma \ref{lemma:dominate_worker} indicates that if there exists a
better worker $w_b$ than worker $w_a$ to do task $t_j$ (in terms of
both the skill set and the traveling cost), then we can safely
filter out the assignment of worker $w_a$ to task $t_j$.

\begin{lemma} (Pruning High-Wage Workers) Let $\widetilde{c_{\cdot j}}$
	be the total traveling cost for those workers that have already been
	assigned to task $t_j$. If the traveling cost $c_{ij}$ of assigning a worker
	$w_i$ to task $t_j$ is greater than the remaining budget
	$(B_j-\widetilde{c_{\cdot j}})$ of task $t_j$, then we will not
	assign worker $w_i$ to task $t_j$. \label{lemma:expensive_worker}
\end{lemma}
\begin{proof}
	Please refer to Appendix C.
\end{proof}

Intuitively, Lemma \ref{lemma:expensive_worker} shows that, if the
wage of a worker $w_i$ (including the traveling cost $c_{ij}$)
exceeds the maximum budget $B_j$ of task $t_j$ (i.e., $c_{ij} >
B_j-\widetilde{c_{\cdot j}}$), then we can safely prune the
worker-and-task assignment pair $\langle w_i, t_j\rangle$.

\noindent\textbf{The Task-Pruning Strategy.} Let $W(t_j)$ be a set of valid
workers that can be assigned to task $t_j$, and $\widetilde{W(t_j)}$
be a set of valid workers that have already been assigned to task
$t_j$. We give the lemma of pruning those tasks with insufficient
budgets below.

\begin{lemma} (Pruning Tasks with Insufficient Budgets) If an unassigned worker $w_i\in
	(W(t_j)-\widetilde{W(t_j)})$ has the highest value of $\frac{\Delta
		S_p}{|X_i \cap (Y_j - \widetilde{Y_j})|}$, and the traveling cost,
	$c_{ij}$, of worker $w_i$ exceeds the remaining budget
	$(B_j-\widetilde{c_{\cdot j}})$ of task $t_j$, then we can safely
	prune task $t_j$. \label{lemma:prune_insufficient_task}
	
\end{lemma}
\begin{proof}
	Please refer to Appendix D.
\end{proof}

Intuitively, Lemma \ref{lemma:prune_insufficient_task} provides the
conditions of pruning tasks. That is, if any unassigned worker
subset of $(W(t_j)-\widetilde{W(t_j)})$ either cannot fully cover
the required skill set $Y_j$, or exceeds the remaining budget of
task $t_j$, then we can directly prune all assignment pairs that
contain task $t_j$.

To summarize, by utilizing Lemmas \ref{lemma:dominate_worker},
\ref{lemma:expensive_worker} and
\ref{lemma:prune_insufficient_task}, we do not have to check all
worker-and-task assignments iteratively in our greedy algorithm.
Instead, we can now apply our proposed three pruning methods, and
effectively filter out those false alarms of assignment pairs, which
can significantly reduce the number of times to compute the score
increases.

\subsection{The Greedy Algorithm}
\label{sub:greedy_algorithm}

According to the definition of the score increase $\Delta S_p$ (as
mentioned in Section \ref{subsec:score_increase}), we propose a
greedy algorithm, which iteratively assigns a worker to a spatial
task that can always achieve the highest score increase.

\begin{figure}[ht]\vspace{-3ex}
	\begin{center}
		\begin{tabular}{l}
			\parbox{3.1in}{
				\begin{scriptsize}
					\begin{tabbing}
						12\=12\=12\=12\=12\=12\=12\=12\=12\=12\=12\=\kill
						{\bf Procedure {\sf MS-SC\_Greedy}} \{ \\
						\> {\bf Input:} $n$ workers in $W_p$ and $m$ time-constrained spatial tasks in $T_p$\\
						\> {\bf Output:} a worker-and-task assignment instance set, $I_p$\\
						\> (1) \> \> $I_p = \emptyset$;\\
						\> (2) \> \> compute all valid worker-and-task pairs $\langle w_i, t_j \rangle$ from $W_p$ and $T_p$\\
						\> (3) \> \> while $W_p \neq \emptyset$ and $T_p \neq \emptyset$\\
						\> (4) \> \> \> $S_{cand} = \emptyset;$\\
						\> (5) \> \> \> for each task $t_j\in T_p$\\
						\> (6) \> \> \> \> for each worker $w_i$ in the valid pair $\langle w_i, t_j \rangle$\\
						\> (7) \> \> \> \> \> if we cannot prune dominated worker $w_i$ by Lemma \ref{lemma:dominate_worker}\\
						\> (8) \> \> \> \> \> \> if we cannot prune high-wage worker $w_i$ by Lemma \ref{lemma:expensive_worker}\\
						\> (9) \> \> \> \> \> \> \> add $\langle w_i, t_j \rangle$ to $S_{cand}$\\
						\> (10)\> \> \> \> if we cannot prune task $t_j$ w.r.t. workers in $S_{cand}$ by Lemma \ref{lemma:prune_insufficient_task}\\
						\> (11)\> \> \> \> \> for each pair $\langle w_i, t_j \rangle$ w.r.t. task $t_j$ in $S_{cand}$\\
						\> (12)\> \> \> \> \> \> compute the score increase, $\Delta S_p(w_i, t_j)$\\
						\> (13)\> \> \> \> else\\
						\> (14)\> \> \> \> \> $T_p=T_p - \{t_j\}$\\
						\> (15)\> \> \> obtain a pair, $\langle w_r, t_j \rangle \in S_{cand}$, with the highest score increase, \\
						\> \> \> \> \> \> $\Delta S_p(w_r, t_j)$, and add this pair to $I_p$\\
						\> (16)\> \> \> $W_p=W_p - \{w_r\}$\\
						\> (17)\> \> return $I_p$\\
						\}
					\end{tabbing}
				\end{scriptsize}
			}
		\end{tabular}
	\end{center}\vspace{-5ex}
	\caption{\small The MS-SC Greedy Algorithm.}\vspace{-4ex}
	\label{alg:greedy}
\end{figure}

Figure \ref{alg:greedy} shows the pseudo code of our MS-SC greedy
algorithm, namely {\sf MS-SC\_Greedy}, which obtains one
worker-and-task pair with the highest score increase each
time, and returns a task assignment instance set $I_p$ with
high score.

Initially, we set $I_p$ to be empty, since no workers are assigned
to any tasks (line 1). Next, we find out all valid worker-and-task
pairs $\langle w_i, t_j \rangle$ in the crowdsourcing system at
timestamp $p$ (line 2). Here, the validity of pair $\langle w_i, t_j
\rangle$ satisfies 4 conditions: (1) the distance between the
current location, $l_i(p)$, of worker $w_i$ and the location, $l_j$
of task $t_j$ is less than the maximum moving distance, $d_i$ of
worker $w_i$, that is, $dist(l_i(p), l_j) \leq d_i$; (2) worker
$w_i$ can arrive at the location, $l_j$, of task $t_j$ before the
arrival deadline $e_j$; (3) worker $w_i$ have skills that task $t_j$
requires; and (4) the traveling cost, $c_{ij}$, of worker $w_i$
should not exceed the budget $B_j$ of task $t_j$.

Then, for each round, we would select one valid worker-and-task
assignment pair with the highest score increase, and add it to set
$I_p$ (lines 3-16). Specifically, in each round, we check every task
$t_j$ that is involved in valid pairs $\langle w_i, t_j \rangle$,
and then prune those dominated and high-wage workers $w_i$, via
Lemmas \ref{lemma:dominate_worker} and \ref{lemma:expensive_worker},
respectively (lines 7-8). If worker $w_i$ cannot be pruned by both
pruning methods, then we add it to a candidate set $S_{cand}$ for
further checking (line 9). After obtaining all workers that match
with task $t_j$, we apply Lemma \ref{lemma:prune_insufficient_task}
to filter out task $t_j$ (if workers cannot be successfully assigned
to $t_j$). If task $t_j$ cannot be pruned, we will calculate the
score increase, $\Delta S_p(w_i, t_j)$, for each pair $\langle w_i,
t_j \rangle$ in $S_{cand}$; otherwise, we remove task $t_j$ from
task set $T_p$ (lines 10-14).

After we scan all tasks in $T_p$, we can retrieve one
worker-and-task assignment pair, $\langle w_r, t_j \rangle$, from
the candidate set $S_{cand}$, which has the highest score increase,
and insert this pair to $I_p$ (line 15). Since worker $w_r$ has been
assigned, we remove it from the worker set $W_p$ (line 16). The
process above repeats, until all workers have been assigned (i.e.,
$W_p = \emptyset$) or there are no tasks left (i.e.,
$T_p=\emptyset$) (line 3).

Figure \ref{subfig:validpairs} illustrates an example of valid
pairs, where $n$ available workers and $m$ spatial tasks are denoted
by triangular and circular nodes, respectively, and valid
worker-and-task pairs are represented by dashed lines. Figure
\ref{subfig:assignment} depicts the result of one assignment with
high score, where the bold lines indicate assignment pairs in $I_p$.

\begin{figure}[ht]\vspace{-2ex}\centering
	\subfigure[][{\scriptsize Valid Pairs}]{
		\scalebox{0.47}[0.47]{\includegraphics{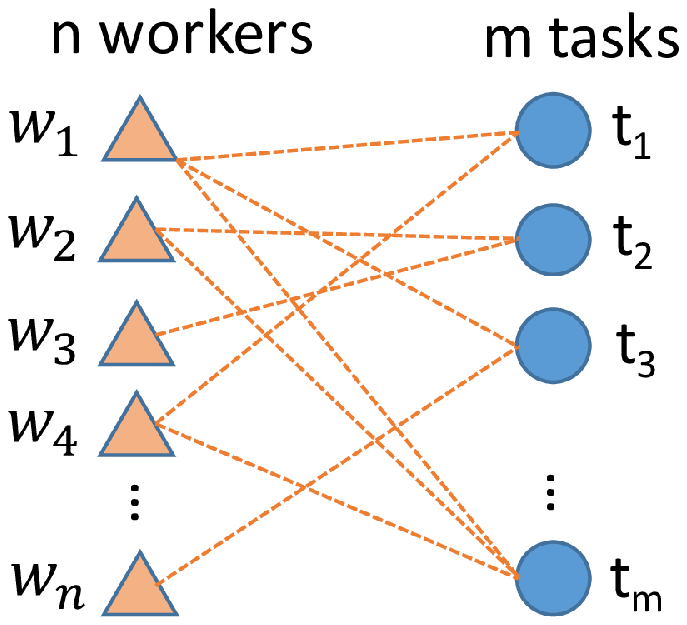}}
		\label{subfig:validpairs}}\hspace{8ex}
	\subfigure[][{\scriptsize Assignment Instance}]{
		\scalebox{0.45}[0.45]{\includegraphics{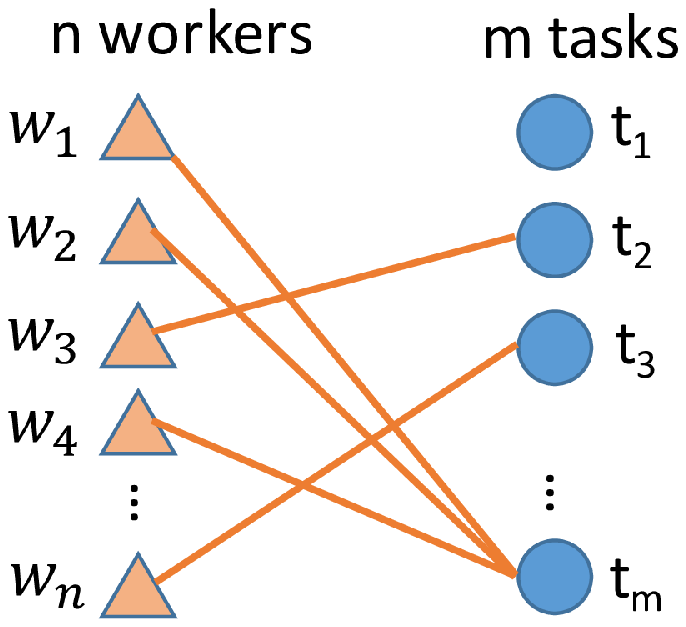}}
		\label{subfig:assignment}}\vspace{-2ex}
	\caption{\small Illustration of the Worker-and-Task
		Assignment.}\vspace{-4ex}
	\label{fig:assignment}
\end{figure}

\vspace{0.5ex}\noindent {\bf The Time Complexity.} We next present
the time complexity of the greedy algorithm, {\sf MS-SC\_Greedy} (in
Figure \ref{alg:greedy}). Specifically, the time cost of computing
valid worker-and-task assignment pairs (line 2) is given by $O(m
\cdot n)$ in the worst case, where any of $n$ workers can be
assigned to any of $m$ tasks (i.e., $m\cdot n$ valid worker-and-task
pairs). Then, for each round (lines 3-16), we apply pruning methods
to $m\cdot n$ pairs, and select the pair with the highest score
increase. In the worst case, pairs cannot be pruned, and thus the
time complexity of computing score increases for these pairs is
given by $O(m \cdot n)$. Moreover, since each of $n$ workers can
only be assigned to one spatial task, the number of iterations is at
most $n$ times. Therefore, the total time complexity of our greedy
algorithm can be given by $O(m \cdot n^2)$.

\section{The $g$-Divide-and-Conquer Approach}
\label{sec:g_D&C}

Although the greedy algorithm incrementally finds one
worker-and-task assignment (with the highest score increase) at a
time, it may incur the problem of only achieving local optimality.
Therefore, in this section, we propose an efficient
\textit{$g$-divide-and-conquer algorithm} ($g$-D\&C), which first
divides the entire MS-SC problem into $g$ subproblems, such that
each subproblem involves a smaller subgroup of $\lceil m/g\rceil$
spatial tasks, and then conquers the subproblems recursively (until
the final group size becomes 1). Since different numbers, $g$, of
the divided subproblems may incur different time costs, in this
paper, we will propose a novel cost-model-based method to estimate
the best $g$ value to divide the problem.

Specifically, for each subproblem/subgroup (containing $\lceil
m/g\rceil$ tasks), we will tackle the worker-and-task assignment
problem via recursion (note: the base case with the group size equal
to 1 can be solved by the greedy algorithm
\cite{vazirani2013approximation}, which has an approximation ratio
of $\ln(N)$, where $N$ is the total number of skills). During the
recursive process, we will combine/merge assignment results from
subgroups, and obtain the assignment strategy for the merged groups,
by resolving the assignment conflict among subgroups. Finally, we
can return the task assignment instance set $I_p$, with respect to
the entire worker and tasks sets.

In the sequel, we first discuss how to decompose the MS-SC problem
into subproblems in Section \ref{subsec:MS_SC_decomposition}. Then,
we will illustrate our $g$-divide-and-conquer approach in Section
\ref{subsec:D&C}, which utilizes the decomposition and merge (as
will be discussed in Section \ref{subsec:merge}) algorithms.
Finally, we will provide a cost model in Section
\ref{subsec:DC_cost_model} to determine the best number $g$ of
subproblems during the $g$-D\&C process.

\begin{figure}[ht]\vspace{-2ex}
	\centering
	\subfigure[][{\scriptsize Original MS-SC Problem}]{
		\scalebox{1}[1]{\includegraphics{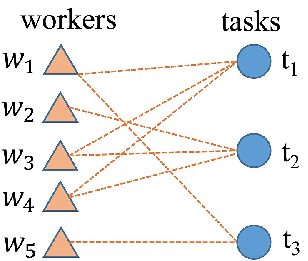}}
		\label{subfig:beforedecomposing}}\hspace{4ex}
	\subfigure[][{\scriptsize Decomposed Subproblems}]{
		\scalebox{0.9}[0.9]{\includegraphics{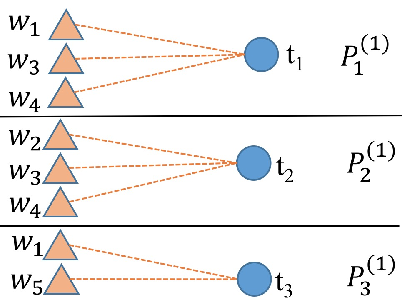}}
		\label{subfig:decomposed}}\vspace{-2ex}
	\caption{\small Illustration of Decomposing the MS-SC Problem.}\vspace{-5ex}
	\label{fig:decomposing}
\end{figure}

\subsection{MS-SC Problem Decompositions}
\label{subsec:MS_SC_decomposition}

In this subsection, we discuss how to decompose a MS-SC problem into
subproblems. In order to illustrate the decomposition, we first
convert our original MS-SC problem into a representation of a
bipartite graph.

\vspace{0.5ex}\noindent {\bf Bipartite Graph Representation of the
	MS-SC Problem.} Specifically, given a worker set $W_p$ and a spatial
task set $T_p$, we denote each worker/task (i.e., $w_i$ or $t_j$) as
a vertex in the bipartite graph, where worker and task vertices have
distinct vertex types. There exists an edge between a worker vertex
$w_i$ and a task vertex $t_j$, if and only if worker $w_i$ can reach
spatial task $t_j$ under the constraints of skills (i.e., $X_i \cap
Y_j \ne \emptyset$), time (i.e., arrival time is before deadline
$e_j$ of arrival), distance (i.e., the traveling distance is below $d_i$), and
budget (i.e., the traveling cost is below task budget $B_j$). We say
that the worker-and-task assignment pair $\langle w_i, t_j\rangle$
is \textit{valid}, if there is an edge between vertices $w_i$ and
$t_j$ in the graph.

As an example in Figure \ref{subfig:beforedecomposing}, we have a
worker set $W_p = \{w_i | 1\leq i\leq 5\}$, and a spatial task set
$T_p = \{t_j | 1\leq j\leq 3\}$, which are denoted by two types of
vertices (i.e., represented by triangle and circle shapes,
respectively) in a bipartite graph. Any edge connects two types of
vertices $w_i$ and $t_j$, if worker $w_i$ can reach the location of
task $t_j$ and do tasks with the required skills from $t_j$. For
example, there exists an edge between $w_1$ and $t_1$, which
indicates that worker $w_1$ can move to the location of $t_1$ before
the arrival deadline $e_1$, with the traveling distance under $d_1$,
with the traveling cost below budget $B_1$, and moreover with some
skill(s) in the required skill set $Y_1$ of task $t_1$.

Note that, one or multiple worker vertices (e.g., $w_1$, $w_3$, and
$w_4$) may be connected to the same task vertex (e.g., $t_1$).
Furthermore, multiple task vertices, say $t_1$ and $t_2$, may also
share some conflicting workers (e.g., $w_3$ or $w_4$), where the
conflicting worker $w_3$ (or $w_4$) can be assigned to either task
$t_1$ or task $t_2$ mutual exclusively.

\begin{figure}[ht]
	\begin{center}
		\begin{tabular}{l}
			\parbox{3.1in}{
				\begin{scriptsize} \vspace{-4ex}
					\begin{tabbing}
						12\=12\=12\=12\=12\=12\=12\=12\=12\=12\=12\=\kill
						{\bf Procedure {\sf MS-SC\_Decomposition}} \{ \\
						\> {\bf Input:} $n$ workers in $W_p$, $m$ time-constrained spatial tasks in $T_p$, and the number\\
						\> \> \> \> of groups $g$\\
						\> {\bf Output:} decomposed MS-SC subproblems, $P_s$ ($1\leq s\leq g$)\\
						\> (1) \> \> for $s$ = 1 to $g$\\
						\> (2) \> \> \> $P_s = \emptyset$\\
						\> (3) \> \> compute all valid worker-and-task pairs $\langle w_i, t_j\rangle$ from $W_p$ and $T_p$, \\
						\> \> \> \> \> and obtain a bipartite graph $G$\\
						\> (4) \> \> for $s = 1$ to $g$\\
						\> (5) \> \> \> let set $T_p^{(j)}$ contain the next anchor task $t_j$ and its top-$(\lceil m/g \rceil -1)$\\
						\>  \> \>  \> \> \> nearest tasks // the task, $t_j$, whose longitude is the smallest\\
						\> (6) \> \> \> for each task vertex $t_j \in T_p^{(j)}$ in graph $G$\\
						\> (7) \> \> \> \> obtain all worker vertices $w_i$ that connect with task vertex $t_j$\\
						\> (8) \> \> \> \> add all pairs $\langle w_i, t_j\rangle$ to $P_s$\\
						\> (9) \> \>  return $P_s$ (for $1\leq s\leq g$)\\
						\}
					\end{tabbing}
				\end{scriptsize}
			}
		\end{tabular}
	\end{center}
	\vspace{-5ex}
	\caption{\small The MS-SC Problem Decomposition Algorithm.}
	\vspace{-4ex}
	\label{alg:decomposing}
\end{figure}

\vspace{0.5ex}\noindent {\bf Decomposing the MS-SC Problem.} Next,
we will illustrate how to decompose the MS-SC problem, with respect
to task vertices in the bipartite graph. Figure
\ref{fig:decomposing} shows an example of decomposing the MS-SC
problem (as shown in Figure \ref{subfig:beforedecomposing}) into 3
subproblems (as depicted in Figure \ref{subfig:decomposed}), where
each subproblem contains a subgroup of one single spatial task
(i.e., group size = 1), associated with its connected worker
vertices. For example, the first subgroup in Figure
\ref{subfig:decomposed}) contains task vertex $t_1$, as well as its
connecting worker vertices $w_1$, $w_3$, and $w_4$. Different task
vertices may have conflicting workers, for example, tasks $t_1$ and
$t_2$ share the same (conflicting) worker vertices $w_3$ and $w_4$.

In a general case, given $n$ workers and $m$ spatial tasks, we
partition the bipartite graph into $g$ subgroups,
each of which contains $\lceil m/g \rceil$ spatial tasks, as well as their
connecting workers. Figure \ref{alg:decomposing} presents the pseudo
code of our MS-SC problem decomposition algorithm, namely {\sf
	MS-SC\_Decomposition}, which returns $g$ MS-SC
subproblems (each corresponding to a subgroup with $\lceil m/g \rceil$ tasks),
$P_s$, after decomposing the original MS-SC problem.

Specifically, we first initialize $g$ empty
subproblems, $P_s$, where $1\leq s\leq g$ (lines
1-2). Then, we find out all valid worker-and-task pairs $\langle
w_i, t_j\rangle$ in the crowdsourcing system at timestamp $p$, which
can can form a bipartite graph $G$, where valid pairs satisfy the
constraints of skills, times, distances, and budgets (line 3).

Next, we want to obtain one subproblem $P_s$ at a time (lines 4-8).
In particular, for each round, we retrieve an anchor task $t_j$ and
its top-$(\lceil m/g \rceil -1)$ nearest tasks, which form a task
set $T_p^{(j)}$ of size $\lceil m/g \rceil$ (line 5). Here, we
choose anchor tasks with a sweeping style, that is, we always choose
the task whose longitude is smallest (in the case where multiple
tasks have the same longitude, we choose the one with smallest
latitude). Then, for each task $t_j\in T_p^{(j)}$, we obtain its
corresponding vertex in $G$ and all of its connecting worker
vertices $w_i$, and add pairs $\langle w_i, t_j\rangle$ to
subproblem $P_s$ (lines 6-8). Finally, we return all the $g$
decomposed subproblems $P_s$.

\begin{figure}[ht]
	\begin{center}
		\begin{tabular}{l}
			\parbox{3.1in}{
				\begin{scriptsize} \vspace{-4ex}
					\begin{tabbing}
						12\=12\=12\=12\=12\=12\=12\=12\=12\=12\=12\=\kill
						{\bf Procedure {\sf MS-SC\_$g$D\&C}} \{ \\
						\> {\bf Input:} $n$ workers in $W_p$, and $m$ time-constrained spatial tasks in $T_p$\\
						\> {\bf Output:} a worker-and-task assignment instance set, $I_p$\\
						\> (1) \> \> $I_p = \emptyset$\\
						\> (2) \> \> estimate the best number of groups, $g$, for $W_p$ and $T_p$\\
						\> (3) \> \> invoke {\sf MS-SC\_Decomposition}$(W_p, T_p, g)$, and obtain subproblems $P_s$\\
						\> (4) \> \> for $s=1$ to $g$\\
						\> (5) \> \> \> if the number of tasks in subproblem $P_s$ (group size) is greater than 1\\
						\> (6) \> \> \> \> $I_p^{(s)}$ = {\sf MS-SC\_$g$D\&C}($W_p (P_s)$, $T_p(P_s)$)\\
						\> (7) \> \> \> else\\
						\> (8) \> \> \> \> invoke classical greedy set cover algorithm to solve subproblem $P_s$,\\
						\> \> \> \> \> \> \>  and obtain assignment results $I_p^{(s)}$\\
						\> (9) \> \> for $i=1$ to $g$\\
						\> (10) \> \>  \> find the next subproblem, $P_s$\\
						\> (11) \> \>  \> $I_p$ = {\sf MS-SC\_Conflict\_Reconcile} ($I_p$, $I_p^{(s)}$) \\
						\> (12) \> \>  return $I_p$\\
						\}
					\end{tabbing}
				\end{scriptsize}
			}
		\end{tabular}
	\end{center}\vspace{-5ex}
	\caption{\small The $g$-Divide-and-Conquer Algorithm.}
	\vspace{-5ex}
	\label{alg:incrementalFrame}
\end{figure}

\subsection{The $g$-D\&C Algorithm}
\label{subsec:D&C}

In this subsection, we propose an efficient
\textit{$g$-divide-and-conquer} ($g$-D\&C) algorithm, namely {\sf
	MS-SC\_$g$D\&C}, which recursively partitions the original MS-SC
problem into subproblems, solves each subproblem (via recursion),
and merges assignment results of subproblems by resolving the
conflicts.

Specifically, in Algorithm {\sf MS-SC\_$g$D\&C}, we first estimate
the best number of groups, $g$, to partition, with respect to $W_p$ and
$T_p$, which is based on the cost model proposed later in Section
\ref{subsec:DC_cost_model} (line 2). Then, we will call the {\sf
	MS-SC\_Decomposition} algorithm (as mentioned in Figure
\ref{alg:decomposing}) to obtain subproblems $P_s$ (line 3). For
each subproblem $P_s$, if $P_s$ involves more than 1 task, then we
can recursively call Algorithm {\sf MS-SC\_$g$D\&C} itself, by
further dividing the subproblem $P_s$ (lines 5-6). Otherwise, when
subproblem $P_s$ contains only one single task, we apply the greedy
algorithm of the classical set cover problem for task set $T_p(P_s)$ and worker set
$W_p(P_s)$ (lines 7-8).

After that, we can obtain an assignment instance set $I_p^{(s)}$ for
each subproblem $P_s$, and merge them into one single
worker-and-task assignment instance set $I_p$, by reconciling the
conflict (lines 9-11). In particular, $I_p$ is initially empty (line
1), and each time merged with an assignment set $I_p^{(s)}$ from
subproblem $P_s$ (lines 10-11). Due to the confliction among
subproblems, we call function {\sf MS-SC\_Conflict\_Reconcile}
$(\cdot, \cdot)$ (discussed later in Section \ref{subsec:merge}) to
resolve the confliction issue during the merging process. Finally,
we can return the merged assignment instance set $I_p$ (line 12).
\vspace{-1ex}

\subsection{Merging Conflict Reconciliation}
\label{subsec:merge}

In this subsection, we introduce the merging conflict reconciliation
procedure, which resolves the conflicts while merging assignment
results of subproblems (i.e., line 11 of Procedure {\sf
	MS-SC\_$g$D\&C}). Assume that $I_p$ is the current assignment
instance set we have merged so far. Given a new subproblem $P_s$
with assignment set $I_p^{(s)}$, Figure \ref{alg:conflict_reconcile}
shows the merging algorithm, namely {\sf MS-SC\_Conflict\_Reconcile},
which combines two assignment sets $I_p$ and $I_p^{(s)}$ by
resolving conflicts.

\begin{figure}[ht]
	\begin{center}
		\begin{tabular}{l}
			\parbox{3.1in}{
				\begin{scriptsize} \vspace{-2ex}
					\begin{tabbing}
						12\=12\=12\=12\=12\=12\=12\=12\=12\=12\=12\=\kill
						{\bf Procedure {\sf MS-SC\_Conflict\_Reconcile}} \{ \\
						\> {\bf Input:} the current assignment instance set, $I_p$, of subproblem $P$ we have merged, \\
						\> \> \> \> and the assignment instance set, $I_p^{(s)}$, of subproblem $P_s$\\
						\> {\bf Output:} a merged worker-and-task assignment instance set, $I_p$\\
						\> (1) \> \> let $W_c$ be a set of all conflicting workers between $I_p$ and $I_p^{(s)}$\\
						\> (2) \> \> while $W_c \neq \emptyset$\\
						\> (3) \> \> \> choose a worker $w_i\in W_c$ with the highest traveling cost in $I_p^{(s)}$\\
						\> (4) \> \> \> if we substitute $w_i$ with $w_i'$ in $P_s$ having the highest score $S_p^{(s)}$\\
						\> (5) \> \> \> \> compute the reduction of the assignment score, $\Delta S_p^{(s)}$\\
						\> (6) \> \> \> if we substitute $w_i$ with $w_i''$ in $P$ having the highest score $S_p$\\
						\> (7) \> \> \> \> compute the reduction of the assignment score, $\Delta S_p$\\
						\> (8) \> \> \> if $\Delta S_p > \Delta S_p^{(s)}$\\
						\> (9) \> \> \> \> substitute worker $w_i$ with $w_i'$ in $I_p^{(s)}$\\
						\> (10)\> \> \> else\\
						\> (11)\> \> \> \> substitute worker $w_i$ with $w_i''$ in $I_p$\\
						\> (12)\> \> \> $W_c = W_c - \{w_i\}$\\
						\> (13)\> \> $I_p = I_p \cup I_p^{(s)}$\\
						\> (14) \> \>  return $I_p$\\
						\}
					\end{tabbing}
				\end{scriptsize}
			}
		\end{tabular}
	\end{center}\vspace{-5ex}
	\caption{\small The Merging Conflict Reconciliation Algorithm.}
	\vspace{-5ex}
	\label{alg:conflict_reconcile}
\end{figure}

In particular, two distinct tasks from two subproblems may be
assigned with the same (conflicting) worker $w_i$. Since each worker
can only be assigned to one spatial task at a time, we thus need to
avoid such a scenario when merging assignment instance sets of two
subproblems (e.g., $I_p$ and $I_p^{(s)}$). Our algorithm in Figure
\ref{alg:conflict_reconcile} first obtain a set, $W_c$, of all
conflicting workers between $I_p$ and $I_p^{(s)}$ (line 1). Then,
each time we greedily solve the conflicts for workers $w_i$ in an
non-decreasing order of the traveling cost (i.e., $c_{ij}$) in
$I^{(s)}_p$ (line 3). Next, in order to resolve the conflicts, we
try to replace worker $w_i$ with another worker $w_i'$ (or $w_i''$)
in $P_{s}$ (or $P$) with the highest score $S_p^{(s)}$ (or
$S_p$), and compute possible reduction of the assignment score,
$\Delta S_p^{(s)}$ (or $\Delta S_p$) (lines 4-7). Note that, here we replace worker $w_i$ with other available workers. If no other workers are available for replacing $w_i$, we may need to sacrifice task $t_j$ that worker $w_i$ is assigned to. For example, when we cannot find another worker to replace
$w_i$ in $P_s$, the substitute of $w_i$ will be set as an empty worker, which means the
assigned task $t_j$ for $w_i$ in $I_p^{(s)}$ will be sacrificed
and $\Delta S_p^{(s)}=B'_j$ (as calculated in Equation \ref{eq:flexible_budget}). In the case that
$\Delta S_p > \Delta S_p^{(s)}$, we substitute worker $w_i$ with
$w_i'$ in $I_p^{(s)}$ (since the replacement of $w_i$ in subproblem
$S_p^{(s)}$ leads to lower score reduction); otherwise, we resolve
conflicts by replacing $w_i$ with $w_i''$ in $I_p$ (lines 8-12). After
resolving all conflicts, we merge assignment instance set $I_p$ with
$I_p^{(s)}$ (line 13), and return the merged result $I_p$.

\subsection{Cost-Model-Based Estimation of the Best Number of Groups}
\label{subsec:DC_cost_model}

In this subsection, we discuss how to estimate the best number of
groups, $g$, such that the total cost of solving the MS-SC problem
in $g$-divide-and-conquer approach is minimized. Specifically, the
cost of the $g$-divide-and-conquer approach consists of 3 parts: the
cost, $F_D$, of decomposing subproblems, that, $F_C$, of conquering
subproblems recursively, and that, $F_M$, of merging subproblems by
resolving conflicts.

Without loss of generality, as illustrated in Figure
\ref{fig:gDC_cost_model}, during the $g$-divide-and-conquer process,
on level $k$, we recursively divide the original MS-SC problem into
$g^k$ subproblems, $P_1^{(k)}$, $P_2^{(k)}$, ..., and
$P_{g^k}^{(k)}$, where each subproblem involves $m/g^k$ spatial
tasks.
\begin{figure}[ht]\vspace{-3ex}
	\centering
	\scalebox{0.4}[0.4]{\includegraphics{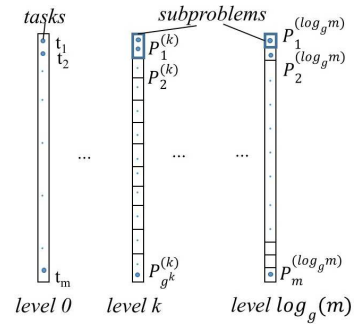}}\vspace{-1ex}
	\caption{\small Illustration of the Cost Model Estimation.}\vspace{-4ex}
	\label{fig:gDC_cost_model}
\end{figure}

\noindent {\bf The Cost, $F_D$, of Decomposing Subproblems.} From
Algorithm {\sf MS-SC\_Decomposition} (in Figure
\ref{alg:decomposing}), we first need to retrieve all valid
worker-and-task assignment pairs (line 3), whose cost is $O(m\cdot
n)$. Then, we will divide each problem into $g$ subproblems, whose
cost is given by $O(m\cdot g+m)$ on each level. For level $k$, we
have $m/g^k$ tasks in each subproblem $P^{(k)}_i$. We will further
divide it into $g$ more subproblems, $P^{(k+1)}_j$, and each one
will have $m/g^{k+1}$ tasks. To obtain $m/g^{k+1}$ tasks in each
subproblem $P^{(k+1)}_j$, we first need to find the anchor task, which
needs $O(m/g^{k})$ cost, and further retrieve the rest tasks,
which needs $O(m/g^{k+1})$ cost. Moreover, since we will have $g^{k+1}$
subproblems on level $k+1$, the cost of decomposing tasks on level $k$
is given by $O(m\cdot g+m)$.

Since there are totally $log_g(m)$ levels, the total cost of
decomposing the MS-SC problem is given by:
{\small
	$$F_D=  m\cdot n+(m \cdot g+m)\cdot log_g(m).$$
	\vspace{-2ex}}

\noindent {\bf The Cost, $F_C$, of Recursively Conquering
	Subproblems.} Let function $F_C(x)$ be the total cost of conquering
a subproblem which contains $x$ spatial tasks. Then, we have the
following recursive function: 
{\small\vspace{-3ex}
	$$F_C(m) = g\cdot F_C(\left\lceil \frac{m}{g}\right\rceil).$$
	\vspace{-2ex}}

Assume that $deg_t$ is the average degree of task nodes in the
bipartite group $G$. Then, the base case of function $F_C(x)$ is the
case that $x=1$, in which we apply the greedy algorithm on just one
single task and $deg_t$ workers. Thus, by the analysis of the time
complexity in Section \ref{sub:greedy_algorithm}, we have:
{\small\vspace{-1ex}
	$$F_C(1) = cost_{greedy}(deg_t, 1) = deg_t^2.$$
	\vspace{-3ex}
}

From the recursive function $F_C(x)$ and its base case, we can
obtain the total cost of the recursive invocation on levels from 1
to $log_g(m)$ below:
{\small\vspace{-1.5ex}
	$$\sum_{k = 1}^{log_g(m)}F_c(m/g^k) =
	\frac{1-m}{1-g}deg_t^2$$
	
}

\noindent {\bf The Cost, $F_M$, of Merging Subproblems.} Next, we
provide the cost, $F_M$, of merging subproblems by resolving
conflicts. Assume that we have $n_s$ workers who could be assigned
to more than one spatial task (i.e., conflicting workers). Moreover,
each worker node has an average degree $deg_w$ in the bipartite
graph. During the subproblem merging processing, we can estimate the
worst-case cost of resolving conflicts for these $n_s$ workers, and
we may resolve conflicts for each worker at most $(deg_w-1)$ times.

Therefore, the worst-case cost of merging subproblems can be given
by: \vspace{-1ex}$$F_M = n_s\cdot (deg_w-1).$$

\noindent {\bf The Total Cost of the $g$-D\&C Approach.} The total
cost, $cost_{gD\&C}$, of the $g$-D\&C algorithm can be given by
summing up three costs, $F_D$, $F_C$, and $F_M$. That is, we have

{\small
	\vspace{-3ex}
	\begin{eqnarray}
		cost_{gD\&C} &=& F_D + \sum_{k = 1}^{log_g(m)}F_c(m/g^k) + F_M \label{eq:D&C_cost}\\
		&=& (mg+m)\log_g (m) + \frac{1-m}{1-g}deg_t^2 + n_s (deg_w -
		1)\notag \vspace{-3ex}
	\end{eqnarray}
	\vspace{-3ex}
}

We take the derivation of $cost_{gD\&C}$ (given in
Eq.~(\ref{eq:D&C_cost})) over $g$, and let it be 0. In particular,
we have:

{\small
	\vspace{-3ex}
	\begin{eqnarray}
		&&\frac{\partial cost_{gD\&C}}{\partial g} \notag\\
		&=&\frac{m\log(m)(g\log(g) - g - 1)}{g\log(2g)} +
		\frac{1-m}{(1-g)^2}deg_t^2 = 0
	\end{eqnarray}
	\vspace{-1ex}
}

We notice that when $g=2$, $\frac{\partial cost_{gD\&C}}{\partial g}$ is much
smaller than 0 but increases quickly when $g$ grows. In addition, $g$ can only
be an integer. Then we can try the integers, (2, 3, 4... ), until
$\frac{\partial cost_{gD\&C}}{\partial g}$ is above 0.

\section{The Cost-Model-Based Adaptive Algorithm}
\label{sec:adpative}

In this section, we introduce a \textit{cost-model-based adaptive
	approach}, which adaptively decides the strategies to apply our
proposed greedy and $g$-divide-and-conquer ($g$-D\&C) algorithms.
The basic idea is as follows. Unlike the $g$-D\&C algorithm, we do
not divide the MS-SC problem into subproblems recursively until task
group sizes become $1$ (which can be solved by the greedy algorithm
of set cover problems). Instead, based on our proposed cost model,
we will partition the problem into subproblems, and adaptively
determine when to stop in some partitioning round (i.e., the total
cost of solving subproblems with the greedy algorithm is smaller
than that of continuing dividing subproblems).

\begin{figure}[ht]
	\begin{center}
		\begin{tabular}{l}
			\parbox{3.1in}{
				\begin{scriptsize} \vspace{-4ex}
					\begin{tabbing}
						12\=12\=12\=12\=12\=12\=12\=12\=12\=12\=12\=\kill
						{\bf Procedure {\sf MS-SC\_Adaptive}} \{ \\
						\> {\bf Input:} $n$ workers in $W_p$, and $m$ time-constrained spatial tasks in $T_p$\\
						\> {\bf Output:} a worker-and-task assignment instance set, $I_p$\\
						\> (1) \> \> $I_p = \emptyset$\\
						\> (2) \> \> estimate the cost, $cost_{greedy}$, of the greedy algorithm\\
						\> (3) \> \> estimate the best number of groups, $g$, and obtain the cost, $cost_{gdc}$, \\
						\> \> \>  of the $g$-D\&C approach\\
						\> (4) \> \> if $cost_{greedy} < cost_{gdc}$\\
						\> (5) \> \> \> $I_p$ = {\sf MS-SC\_Greedy}($W_p$, $T_p$) \\
						\> (6) \> \> else \qquad \textit{// $g$-D\&C algorithm}\\
						\> (7) \> \> \>  invoke {\sf MS-SC\_Decomposition}$(W_p, T_p, g)$, and obtain subproblems $P_s$\\
						\> (8) \> \> \> for each subproblem, $P_s$, \\
						\> (9) \> \> \> \> $I_p^{(s)}$ = {\sf MS-SC\_Adaptive}($W_p(P_s)$, $T_p(P_s)$)\\
						\> (10)\> \> \> for $i=1$ to $g$\\
						\> (11)\> \>  \> \> find the next subproblem, $P_s$\\
						\> (12)\> \> \> \> $I_p$ = {\sf MS-SC\_Conflict\_Reconcile} ($I_p$, $I_p^{(s)}$) \\
						\> (13)\> \>  return $I_p$\\
						\}
					\end{tabbing}
				\end{scriptsize}
			}
		\end{tabular}
	\end{center}\vspace{-5ex}
	\caption{\small The MS-SC Cost-Model-Based Adaptive Algorithm.}
	\label{alg:hybrid}\vspace{-5ex}
\end{figure}

\subsection{Algorithm of the Cost-Model-Based Adaptive Approach}

Figure \ref{alg:hybrid} shows the pseudo-code of our
cost-model-based adaptive algorithm, namely {\sf MS-SC\_Adaptive}.
Initially, we estimate the cost, $cost_{greedy}$, of applying the
greedy approach over worker/task sets $W_p$ and $T_p$ (line 2).
Similarly, we also estimate the best group size, $g$, and compute
the cost, $cost_{gd\&c}$ of using the $g$-D\&C algorithm (line 3).
If it holds that the cost of the greedy algorithm is smaller than
that of the $g$-D\&C approach (i.e., $cost_{greedy} < cost_{gdc}$),
then we will use the greedy algorithm by invoking function {\sf
	MS-SC\_Greedy}($\cdot, \cdot$) (due to its lower cost; lines 4-5).
Otherwise, we will apply the $g$-D\&C algorithm, and further
partition the problem into subproblems $P_s$ (lines 6-7). Then, for
each subproblem $P_s$, we recursively call the cost-model-based
adaptive algorithm, and retrieve the assignment instance set
$I_p^{(s)}$ (line 9). After that, we merge all the assignment
instance sets from subproblems by invoking function {\sf
	MS-SC\_Conflict\_Reconcile}($\cdot, \cdot$) (lines 10-12). Finally,
we return the worker-and-task assignment instance set $I_p$ (line
13).

\subsection{Cost Model for the Stoping Condition}

Next, we discuss how to determine the stopping level, when using our
cost-model-based adaptive approach to recursively solve the MS-SC
problem. Intuitively, at the current level $k$, we need to estimate
the costs, $cost_{greedy}$ and $cost_{gdc}$, of using greedy and
$g$-D\&C algorithms, respectively, to solve the remaining MS-SC
problem. If the greedy algorithm has lower cost, then we will stop
the divide-and-conquer, and apply the greedy algorithm for each
subproblems.

In the sequel, we discuss how to obtain the formulae of costs
$cost_{greedy}$ and $cost_{gdc}$.

\underline{\it The Cost, $cost_{greedy}$, of the Greedy Algorithm.}
Given a set, $W_p$, of $n$ workers and a set, $T_p$, of $m$ tasks,
the cost, $cost_{greedy}$, of our greedy approach (as given in
Figure \ref{alg:greedy}) has been discussed in Section
\ref{sub:greedy_algorithm}.

In the bipartite graph of valid worker-and-task pairs, denote
the average degree of workers as $deg_w$, and that of tasks as
$deg_t$. In Figure \ref{alg:greedy}, the computation of valid
worker-and-task pairs in line 2 needs
$O(m\cdot n)$ cost. Since there are at most $n$ iterations, for each
round (lines 3-16), we apply two worker-pruning methods to at most
$(2m \cdot deg_t)$ pairs, and select pairs with the highest score
increases, which need $O(3m \cdot n\cdot deg_t)$ cost in total.
For the cost of task-pruning, there are totally $n$ rounds (lines 3-16; 
i.e., removing one out of $n$ workers in each round in line 16). 
In each round, there are at most $deg_w$ out of $m$ tasks (line 5) that may be 
potentially pruned by Lemma \ref{lemma:prune_insufficient_task} (line10). To check each of $deg_w$ tasks, we need $O(deg_t)$ cost. 
Therefore, the total cost of task-pruning is given by $O(n\cdot deg_t \cdot deg_w)$. 
If we cannot prune a task that was assigned with a worker in the last round (lines 3-16), then we need to update score increases of $deg_t$ workers for that task. Each task will be assigned with workers for $deg_t$  times. Thus, the total update cost for one task is given by $O(deg^2_t)$ (line 12). Therefore, $cost_{greedy}(n, m)$ can be given by:\vspace{-3ex}

{\small
	\begin{eqnarray}
		&& cost_{greedy}(n, m) \notag\\
		&=&C_{greedy}\cdot(m\cdot n + n\cdot deg_t \cdot (3m + deg_w) + m\cdot deg^2_t), \qquad \label{eq:greedy_approach_cost}
	\end{eqnarray}}
	\vspace{-3ex}
	
	\noindent where parameter $C_{greedy}$ is a constant factor, which
	can be inferred from cost statistics of the greedy algorithm.

	\underline{\it The Cost, $cost_{gdc}$, of the $g$-D\&C Algorithm.}
	Assume that the current $g$-divide-and-conquer level is $k$. We can
	modify the cost analysis in Section \ref{subsec:DC_cost_model}, by
	considering the cost, $cost_{gdc}$, of the remaining
	divide-and-conquer levels. Specifically, we have the cost, $F_D'$,
	of the decomposition algorithm, that is:\vspace{-2ex}
	
	{\small
		$$F_D'=  m\cdot n+(m\cdot g+m)\cdot k.$$
		\vspace{-3ex}}

	Moreover, when the current level is $k$, the cost of conquering the
	remaining subproblems is given by:\vspace{-2ex}
	
	{\small
		$$\sum_{i= k}^{log_g(m)}F_c(m/g^i).$$
		\vspace{-3ex}}
	
	Finally, the cost of merging subproblems is given by $F_M$.
	
	As a result, the total cost, $cost_{gdc}$, of solving the MS-SC
	problem with our $g$-D\&C approach for the remaining partitioning
	levels (from level $k$ to $log_g(m)$) can be given by:\vspace{-3ex}
	
	{\small
		$$cost_{gdc}= C_{gdc}\cdot(F_D' + \sum_{i = k}^{log_g(m)}F_c(m/g^i)+ F_M),$$
		\vspace{-3ex}}
	
	\noindent where parameter $C_{gdc}$ is a constant factor, which can
	be inferred from time cost statistics of the $g$-D\&C algorithm.
	
	This way, we compare $cost_{greedy}$ with $cost_{gdc}$ (as mentioned
	in line 4 of {\sf MS-SC\_Adaptive} Algorithm). If $cost_{greedy}$ is
	smaller than $cost_{gdc}$, we stop at the current level $k$, and apply the greedy algorithm to tackle
	the MS-SC problem directly; otherwise, we keep dividing the original
	MS-SC problem into subproblems (i.e., $g$-D\&C).

\section{Experimental Study}
\label{sec:exper}

\subsection{Experimental Methodology}

\noindent \textbf{Data Sets.} We use both real and synthetic data to
test our proposed MS-SC approaches. Specifically, for real data, we
use Meetup data set from \cite{liu2012event}, which was crawled from
\textit{meetup.com} between Oct. 2011 and Jan. 2012. There are 5,153,886
users, 5,183,840 events, and 97,587 groups in Meetup, where each
user is associated with a location and a set of tags, each group is
associated with a set of tags, and each event is associated with a
location and a group who created the event. For an event, we use the
tags of the group who creates the event as its tags. To conduct the
experiments on our approaches, we use the locations and tags of
users in Meetup to initialize the locations and the practiced skills
of workers in our MS-SC problem. In addition, we utilize the
locations and tags of events to initialize the locations and the
required skills of tasks in our experiments. Since workers are
unlikely to move between two distant cities to conduct one spatial
task, and the constraints of time (i.e., $e_j$), budget (i.e.,
$B_j$) and distance (i.e., $d_i$) also prevent workers from moving
too far, we only consider those user-and-event pairs located in the
same city. Specifically, we select one famous and popular city, Hong
Kong, and extract Meetup records from the area of Hong Kong (with
latitude from \ang{22.209} to \ang{113.843} and longitude from
\ang{22.609} to \ang{114.283}), in which we obtain 1,282 tasks and
3,525 workers.

For synthetic data, we generate locations of workers and tasks in a
2D data space $[0, 1]^2$, following either Uniform (UNIFORM) or
Skewed (SKEWED) distribution. For Uniform distribution, 
we uniformly generate the locations of tasks/workers in the 2D data space.
Similarly, we also generate
tasks/workers with the Skewed distribution by locating 90\% of them into
a Gaussian cluster (centered at (0.5, 0.5) with variance = $0.2^2$),
and distribute the rest workers/tasks uniformly.  Then, for
skills of each worker, we randomly associate one user in Meetup data
set to this worker, and use tags of the user as his/her skills in
our MS-SC system. For the required skills of each task, we randomly 
select an event, and use its tags as the required skills of the task.

For both real and synthetic data sets, we simulate the velocity of
each worker with Gaussian distribution within range $[v^-, v^+]$,
for $v^-, v^+ \in (0, 1)$. For the unit price, $C_i$, w.r.t. the
traveling distance of each worker, we generate it following the Uniform
distribution within the range $[C^-, C^+]$. Furthermore, we produce
the maximum moving distance of each worker, following the Uniform
distribution within the range $[d^-, d^+]$ (for $d^-, d^+ \in
(0,1)$). For temporal constraints of tasks, we also generate the
arrival deadlines of tasks, $e$, within range $[rt^-, rt^+]$ with
Gaussian distribution. Finally, we set the budgets of tasks with
Gaussian distribution within the range $[B^-, B^+]$. 
Here, for Gaussian distributions, we linearly map data 
samples within $[-1, 1]$ of a Gaussian distribution $\mathcal{N}(0, 0.2^2)$ 
to the target ranges.

\noindent\textbf{MS-SC Approaches and Measures.} We conduct
experiments to compare our three approaches, GREEDY, $g$-D\&C and
ADAPTIVE, with a random method, namely RANDOM, which randomly
assigns workers to tasks.

In particular, GREEDY selects a ``best'' worker-and-task assignment
with the highest score increase each time, which is a local optimal
approach. The $g$-D\&C algorithm keeps dividing the problem into $g$
subproblems on each level, until finally the number of tasks in each
subproblem is 1 (which can be solved by the greedy algorithm on each
one-task subproblem). Here, the parameter $g$ can be estimated by a
cost model to minimize the computing cost. The cost-model-base
adaptive algorithm (ADAPTIVE) makes the trade-off between GREEDY and
$g$-D\&C, in terms of efficiency and accuracy, which adaptively
decides the stopping level of the divide-and-conquer. To evaluate
our three proposed approaches, we need to compare the results with
ground truth. However, as proved in Section \ref{sec:reduction}, the
MS-SC problem is NP-hard, and thus infeasible to calculate the real
optimal result as the ground truth. Alternatively, we will compare
the effectiveness of our three approaches with that of a random
(RANDOM) method, which randomly chooses a task then randomly assigns
worker to the task. For each MS-SC instance, we run RANDOM for 10
times, and report the result with the highest score.

Table \ref{table2} depicts our experimental settings, where the
default values of parameters are in bold font. In each set of
experiments, we vary one parameter, while setting other parameters
to their default values. For each experiment, we report the running
time and the assignment score of our tested approaches. All our
experiments were run on an Intel Xeon X5675 CPU @3.07 GHZ with 32 GB
RAM in Java.

\begin{table}[t]
	\begin{center}\vspace{-6ex}
		\caption{\small Experiments Settings.} \label{table2}
		{\small\scriptsize
			\begin{tabular}{l|l}
				{\bf \qquad \qquad \quad Parameters} & {\bf \qquad \qquad \qquad Values} \\ \hline \hline
				the number of tasks $m$ & 1K, 2K, \textbf{5K}, 8K, 10K \\
				the number of workers $n$ & 1K, 2K, \textbf{5K}, 8K, 10K \\
				the task budget range $[B^-, B^+]$  & [1, 5], \textbf{[5, 10]}, [10, 15], [15, 20], [20, 25]\\
				the velocity range $[v^-, v^+]$   & [0.1, 0.2], \textbf{[0.2, 0.3]}, [0.3, 0.4], [0.4, 0.5]\\
				the unit price w.r.t. distance $[C^-, C^+]$ & [10, 20], \textbf{[20, 30]}, [30, 40], [40, 50]\\
				the moving distance range $[d^-, d^+]$ & [0.1, 0.2], [0.2, 0.3], \textbf{[0.3, 0.4]}, [0.4, 0.5]\\
				the expiration time range $[rt^-, rt^+]$  & [0.25, 0.5], [0.5, 1], \textbf{[1, 2]}, [2, 3], [3, 4]\\
				\hline
			\end{tabular}
		}
	\end{center}\vspace{-7ex}
\end{table}

\subsection{Experiments on Real Data}

In this subsection, we show the effects of the range of task budgets
$[B^-, B^+]$, the range of workers' velocities $[v^-, v^+]$, and
the range of unit prices w.r.t. distance $[C^-, C^+]$.

\noindent\textbf{Effect of the Range of Task Budgets $[B^-, B^+]$.}
Figure \ref{fig:budget_b} illustrates the experimental results on
different ranges, $[B^-, B^+]$, of task budgets $B_j$ from $[1,5]$
to $[20,25]$. In Figure \ref{subfig:b_score}, the assignment scores
of all the four approaches increase, when the value range of task
budgets gets larger. When the average budgets of tasks increase, the
flexible budget $B'$ of each task will also increase. $g$-D\&C and
ADAPTIVE can achieve higher score than GREEDY. In contrast, RANDOM
has the lowest score, which shows that our proposed three approaches
are more effective. As shown in Figure \ref{subfig:b_cpu}, the
running times of our three approaches increase, when the range of
task budgets becomes larger. The reason is that, when $B_j \in [B^-,
B^+]$ increases, each task has more valid workers, which thus leads
to higher complexity of the MS-SC problem and the increase of the
running time. The RANDOM approach is the fastest (however, with the
lowest assignment score), since it does not even need to find local
optimal assignment. The ADAPTIVE algorithm achieves much lower
running time than $g$-D\&C (a bit higher time cost than GREEDY), but
has comparable score with $g$-D\&C (much higher score than GREEDY),
which shows the good performance of ADAPTIVE, compared with GREEDY
and $g$-D\&C.

\begin{figure}[ht]\vspace{-2ex}
	\centering
	\subfigure[][{\scriptsize Scores of Assignment}]{
		\scalebox{0.2}[0.2]{\includegraphics{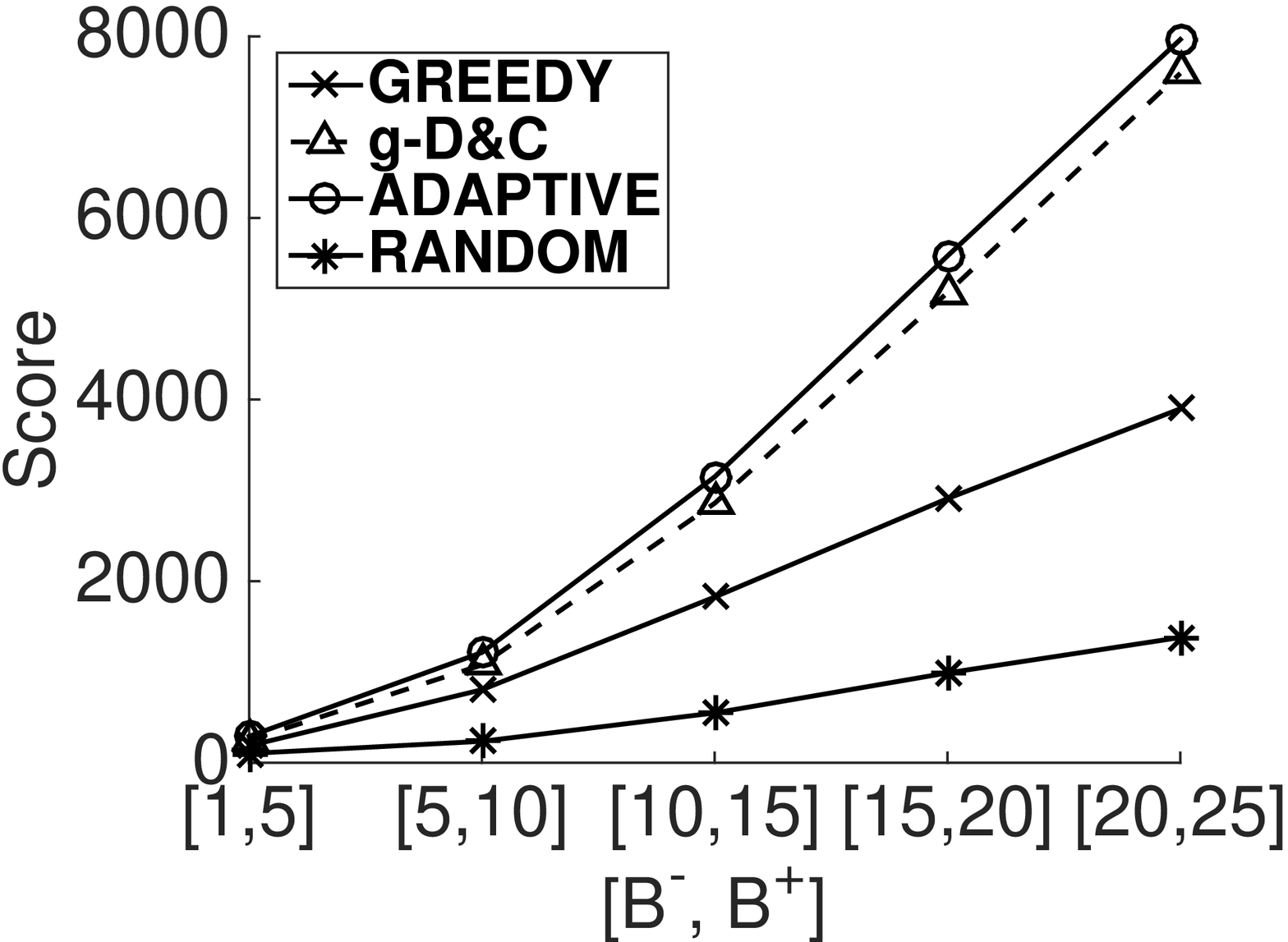}}
		\label{subfig:b_score}}
	\subfigure[][{\scriptsize Running Times}]{
		\scalebox{0.2}[0.2]{\includegraphics{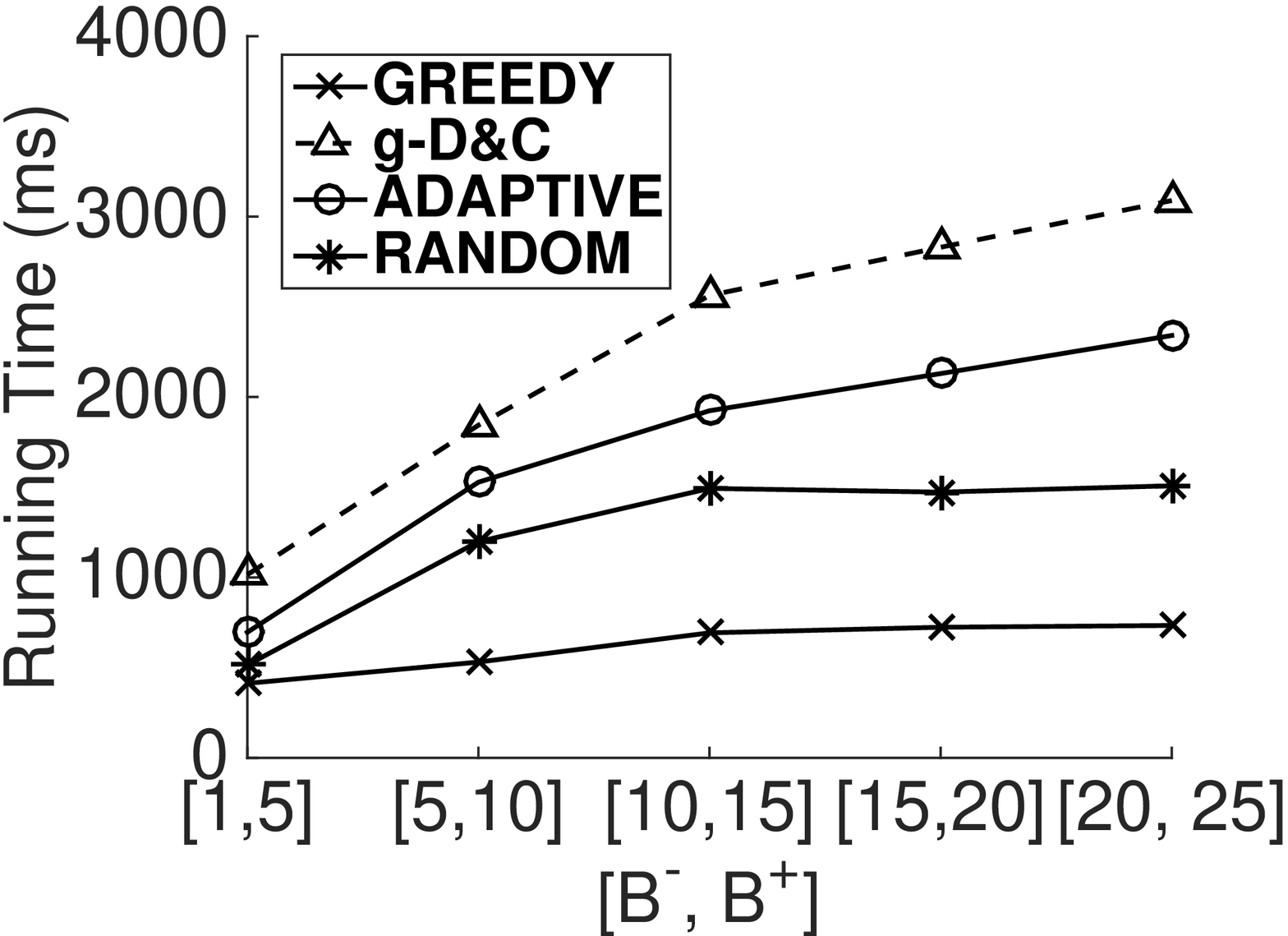}}
		\label{subfig:b_cpu}}\vspace{-2ex}
	\caption{\small Effect of the Range of Task Budgets $[B^-, B^+]$ (Real Data).}\vspace{-4ex}
	\label{fig:budget_b}
\end{figure}

\noindent\textbf{Effect of the Workers' Velocity Range $[v^-,
	v^+]$.} Figure \ref{fig:velocity_v} reports the effect of the range
of velocities, $[v^-, v^+]$, of workers over real data. As shown in
Figure \ref{subfig:v_score}, when the range of velocities increases
from $[0.1,0.2]$ to $[0.2,0.3]$, the scores of all the approaches
first increase; then, they stop growing for the velocity range 
varying from [0.2, 0.3] to [0.4, 0.5]. The reason is that, at the
beginning, with the increase of velocities, workers can reach more
tasks before their arrival deadlines. Nevertheless, workers are also
constrained by their maximum moving distances, which prevents them
from reaching more tasks. ADAPTIVE can achieve a bit higher scores
than $g$-D\&C, and much better assignment scores than GREEDY.

In Figure \ref{subfig:v_cpu}, when the range of velocity $[v^-,
v^+]$ increases, the running times of our tested approaches also
increase, due to the cost of more valid worker-and-task pairs to be
handled. Similar to previous results, RANDOM is the fastest, and
$g$-D\&C is the slowest. ADAPTIVE requires about 0.5-1.5 seconds,
and has lower time cost than $g$-D\&C, which shows the efficiency of
our proposed approaches.

\begin{figure}[ht]\vspace{-2ex}
	\centering
	\subfigure[][{\scriptsize Scores of Assignment}]{
		\scalebox{0.2}[0.2]{\includegraphics{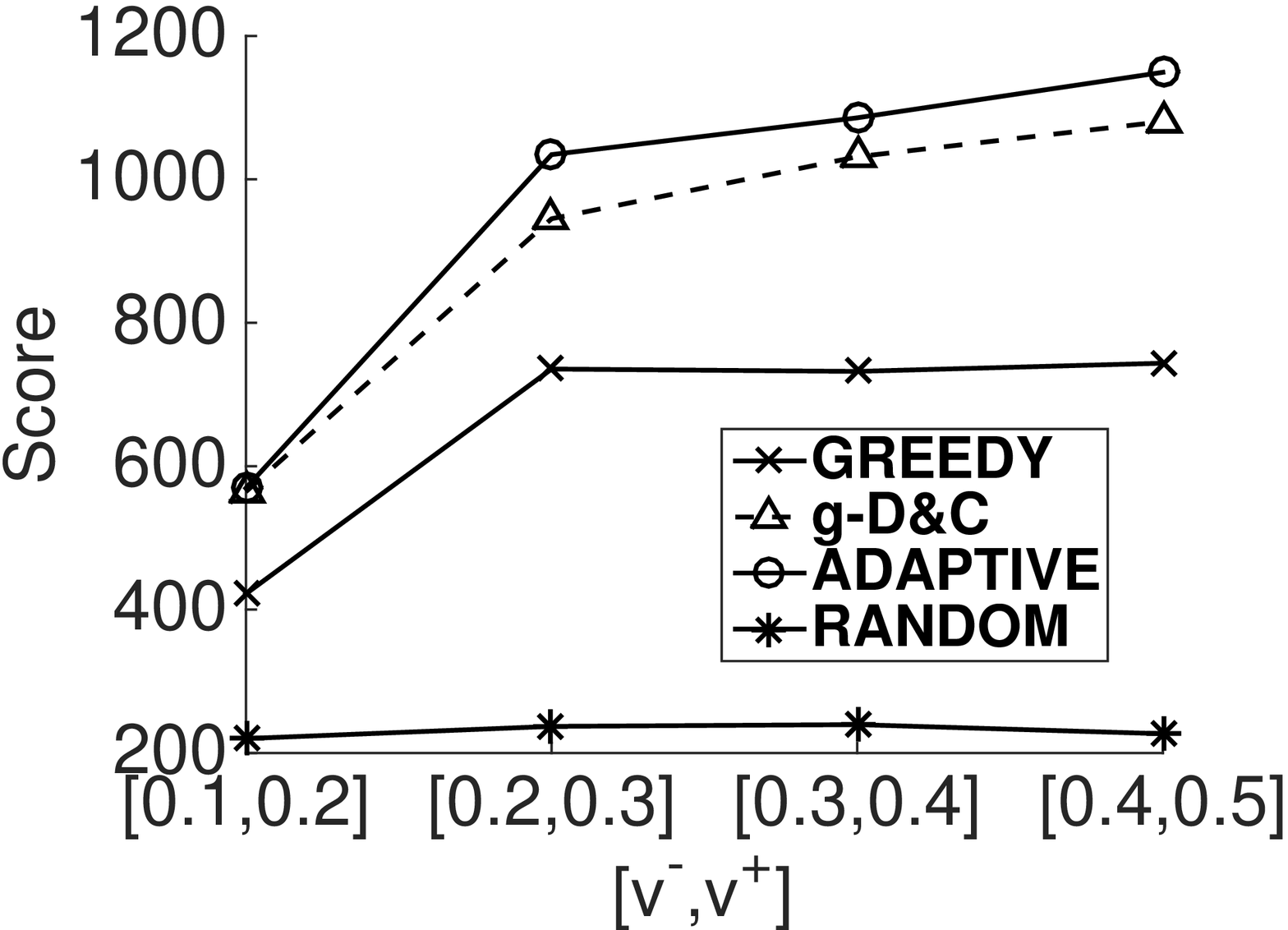}}
		\label{subfig:v_score}}
	\subfigure[][{\scriptsize Running Times}]{
		\scalebox{0.2}[0.2]{\includegraphics{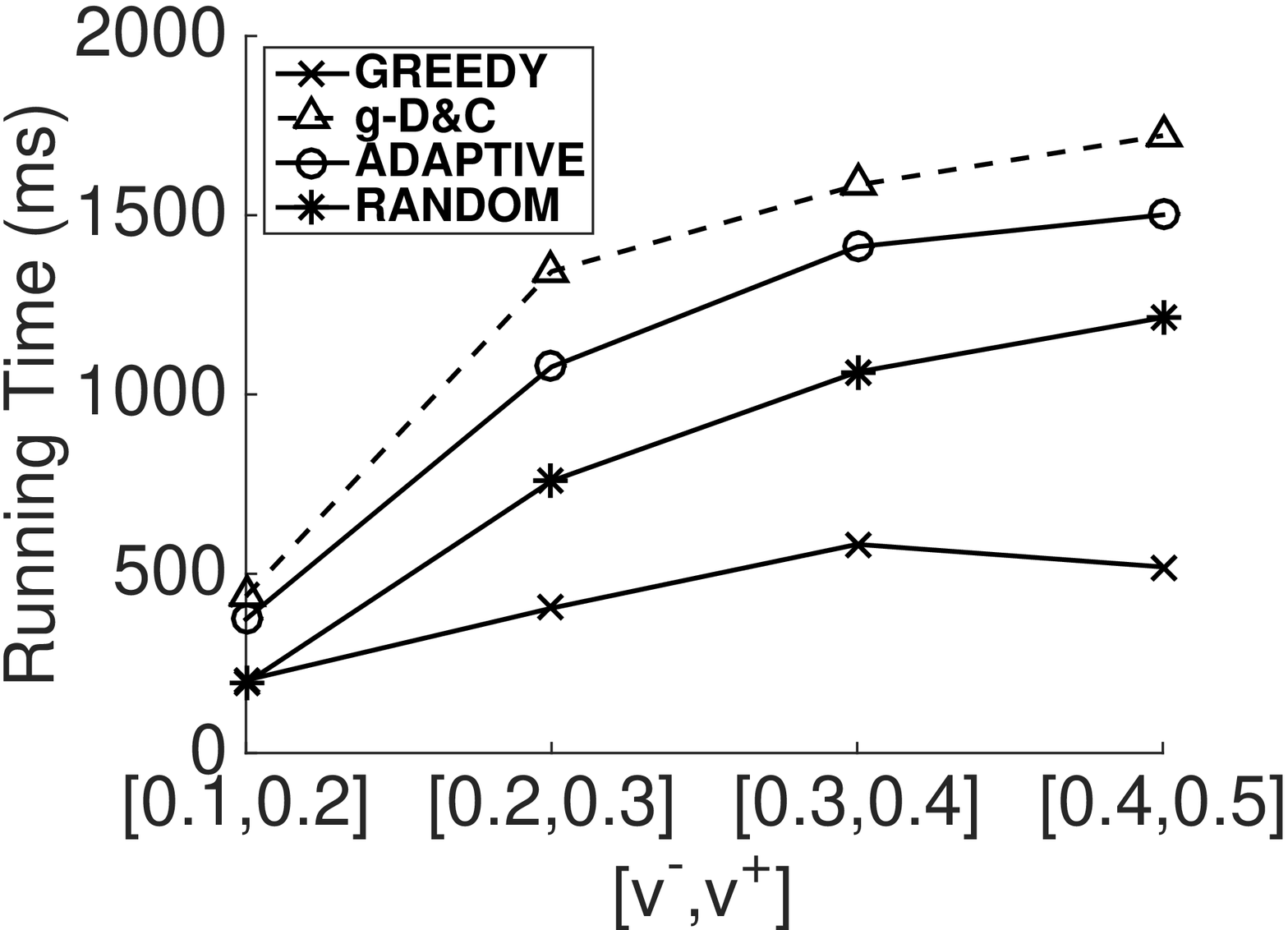}}
		\label{subfig:v_cpu}}\vspace{-2ex}
	\caption{\small Effect of the Range of Velocities $[v^-, v^+]$ (Real Data).}\vspace{-4ex}
	\label{fig:velocity_v}
\end{figure}

\noindent\textbf{Effect of the Range of Unit Prices w.r.t. Traveling Distance
	$[C^-, C^+]$.} In Figure \ref{subfig:c_score}, when the unit prices w.r.t.
traveling distance $C_i \in[C^-, C^+]$ increase, the scores of all the approaches
decrease. The reason is that, when the range of unit prices $[C^-, C^+]$ increases, we need to pay
more wages containing the traveling costs of workers (to do spatial
tasks), which in turn decreases the flexible budget of each task.
However, ADAPTIVE can still achieve the highest score among all four
approaches; scores of $g$-D\&C are close to the scores of ADAPTIVE
and higher than GREEDY.

In Figure \ref{subfig:c_cpu}, when the range of unit prices, $[C^-, C^+]$, of the
traveling cost increases, the number of valid worker-and-task pairs
decreases, and thus the running time of all the approaches also
decreases. Our ADAPTIVE algorithm is faster than $g$-D\&C, and
slower than GREEDY.

\begin{figure}[ht]\vspace{-2ex}
	\centering
	\subfigure[][{\scriptsize Scores of Assignment}]{
		\scalebox{0.2}[0.2]{\includegraphics{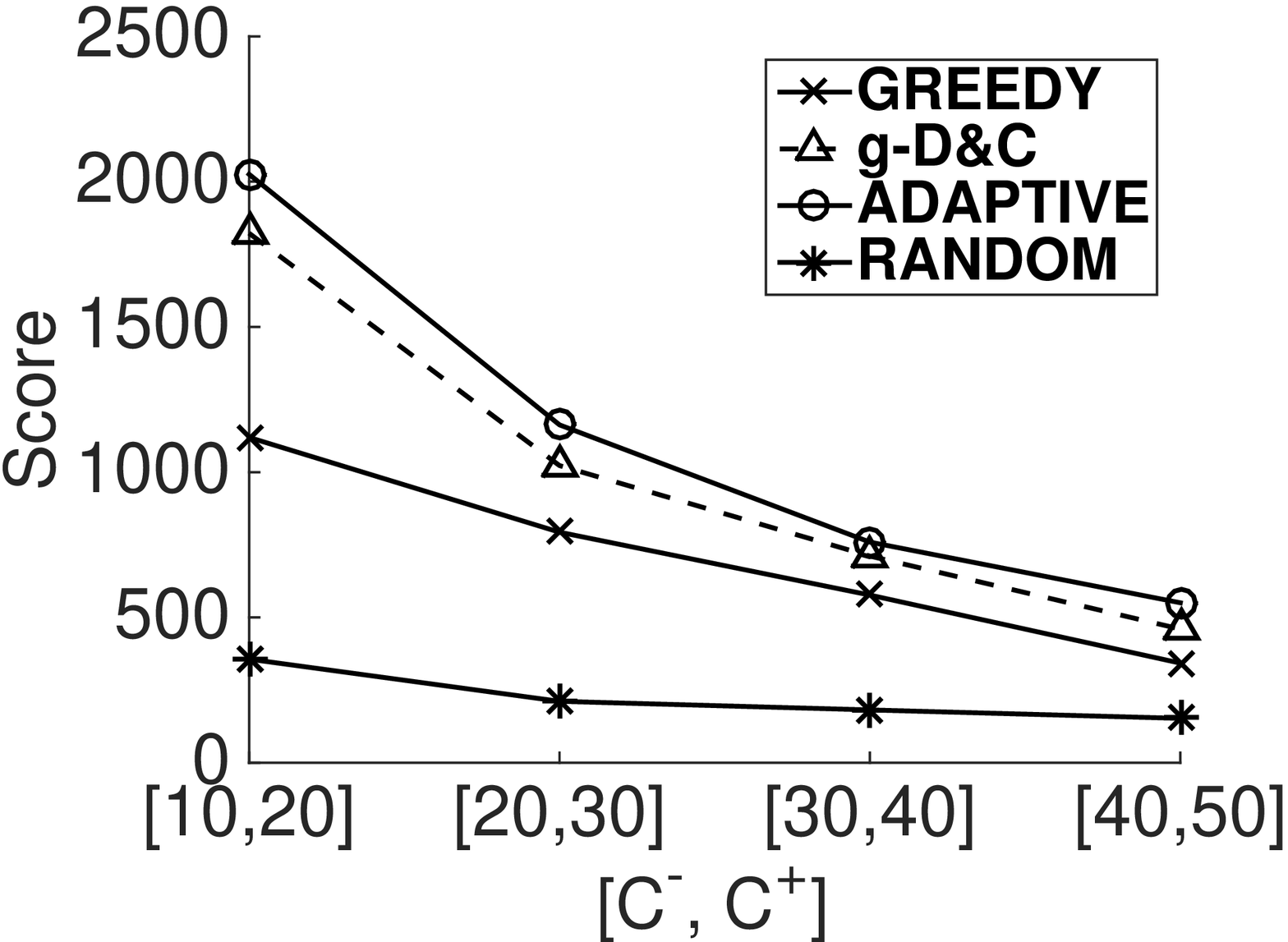}}
		\label{subfig:c_score}}
	\subfigure[][{\scriptsize Running Times}]{
		\scalebox{0.2}[0.2]{\includegraphics{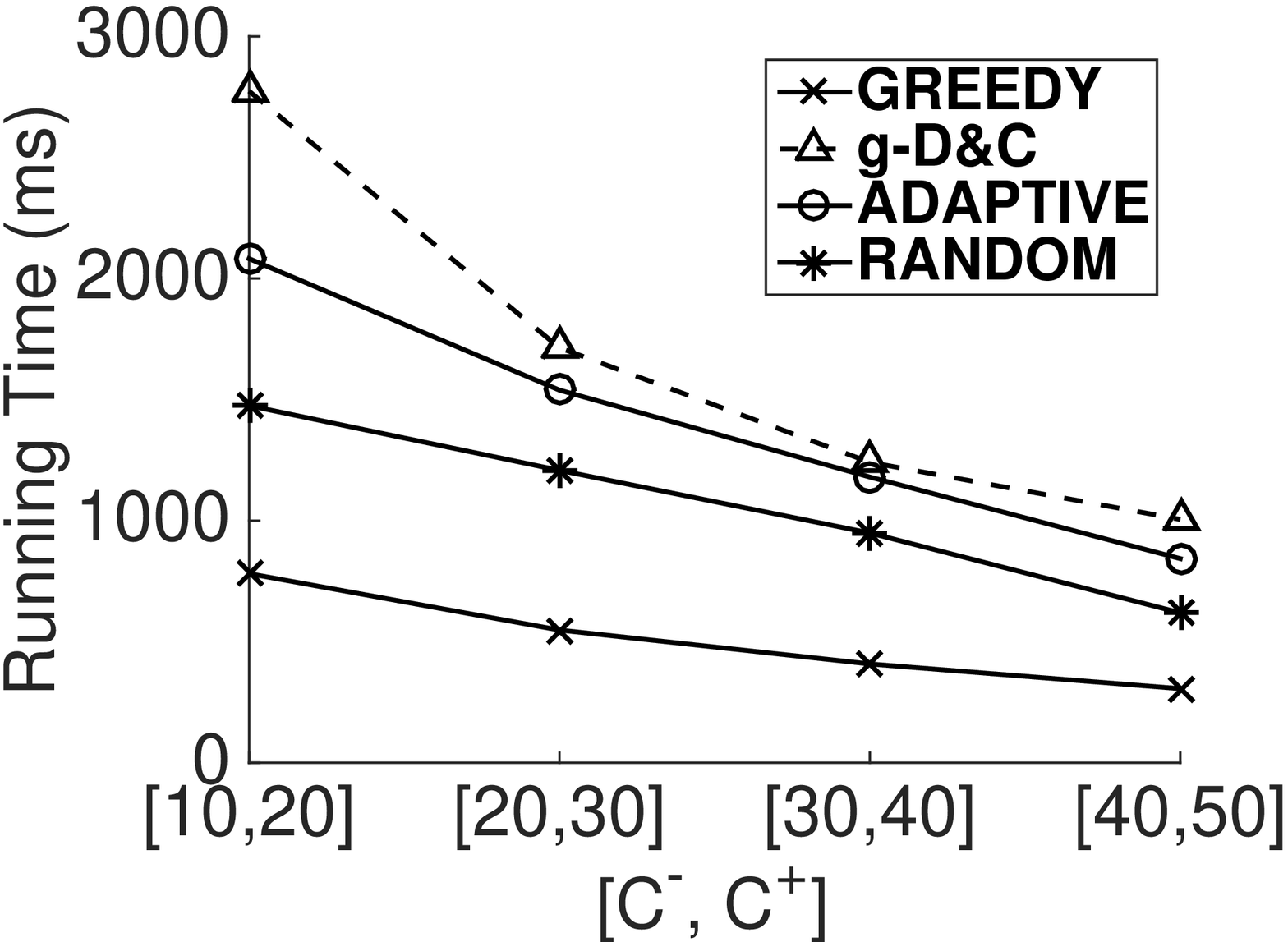}}
		\label{subfig:c_cpu}}\vspace{-2ex}
	\caption{\small Effect of the Range of Unit Prices w.r.t. Traveling Distance $[C^-, C^+]$ (Real Data).} \vspace{-1ex}
	\label{fig:constanct_c}\vspace{-3ex}
\end{figure}

In addition, we also tested the effects of the range, $[d^-, d^+]$,
of maximum moving distances for workers, and the expiration time
range, $[rt^-, rt^+]$, of tasks over the real data set, Meetup. Due to
space limitations, please refer to experimental results with respect
to other parameters (e.g., $[d^-, d^+]$ and $[rt^-, rt^+]$) in
Appendix F.

From experimental results on the real data above, ADAPTIVE can
achieve higher scores than Greedy and $g$-D\&C, and it is faster
than $g$-D\&C and slower than GREEDY. Although $g$-D\&C can achieve
good scores close to ADAPTIVE, it is the slowest among all the 4
approaches.

\subsection{Experiments on Synthetic Data}

In this subsection, we test the effectiveness and robustness of our
three MS-SC approaches, GREEDY, $g$-D\&C, and ADAPTIVE, compared
with RANDOM, by varying parameters (e.g., the number of tasks $m$
and the number of workers $n$) on synthetic data sets. Due to space limitations, we present the experimental results for tasks/workers with Uniform distributions. For similar results with tasks/workers following skewed distributions, please refer to Appendix G.

\noindent\textbf{Effect of the Number of Tasks $m$.} Figure
\ref{fig:tasks_m} illustrates the effect of the number, $m$, of
spatial tasks, by varying $m$ from $1K$ to $10K$, over synthetic
data sets, where other parameters are set to default values. For
assignment scores in Figure \ref{subfig:m_score}, $g$-D\&C obtains
results with the highest scores among all the four approaches.
ADAPTIVE performs similar to $g$-D\&C, and achieves good results
similar to $g$-D\&C. GREEDY is not as good as $g$-D\&C and ADAPTIVE,
but is still much better than RANDOM. When the number, $m$, of
spatial tasks becomes larger, all our approaches can achieve higher
scores.

In Figure \ref{subfig:m_cpu}, when $m$ increases, the running time
also increases. This is because, we need to deal with more
worker-and-task assignment pairs for large $m$. The ADAPTIVE
algorithm is slower than GREEDY, and faster than $g$-D\&C. In
addition, we find that the running time of GREEDY performs, with the
same trend as that estimated in our cost model (as given in
Eq.~(\ref{eq:greedy_approach_cost})).

\begin{figure}[ht]\vspace{-2ex}
	\centering
	\subfigure[][{\scriptsize Scores of Assignment}]{
		\scalebox{0.2}[0.2]{\includegraphics{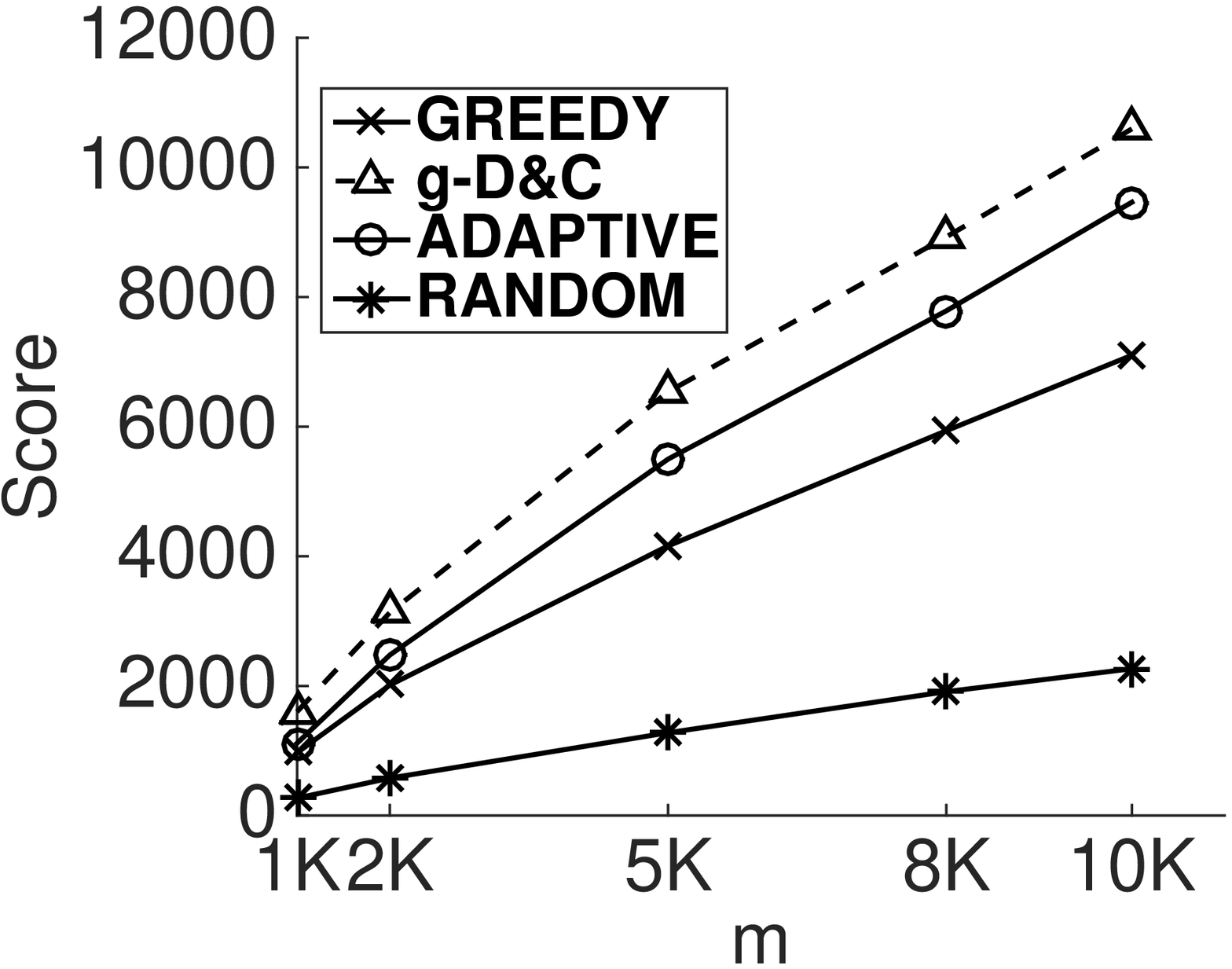}}
		\label{subfig:m_score}}
	\subfigure[][{\scriptsize Running Times}]{
		\scalebox{0.2}[0.2]{\includegraphics{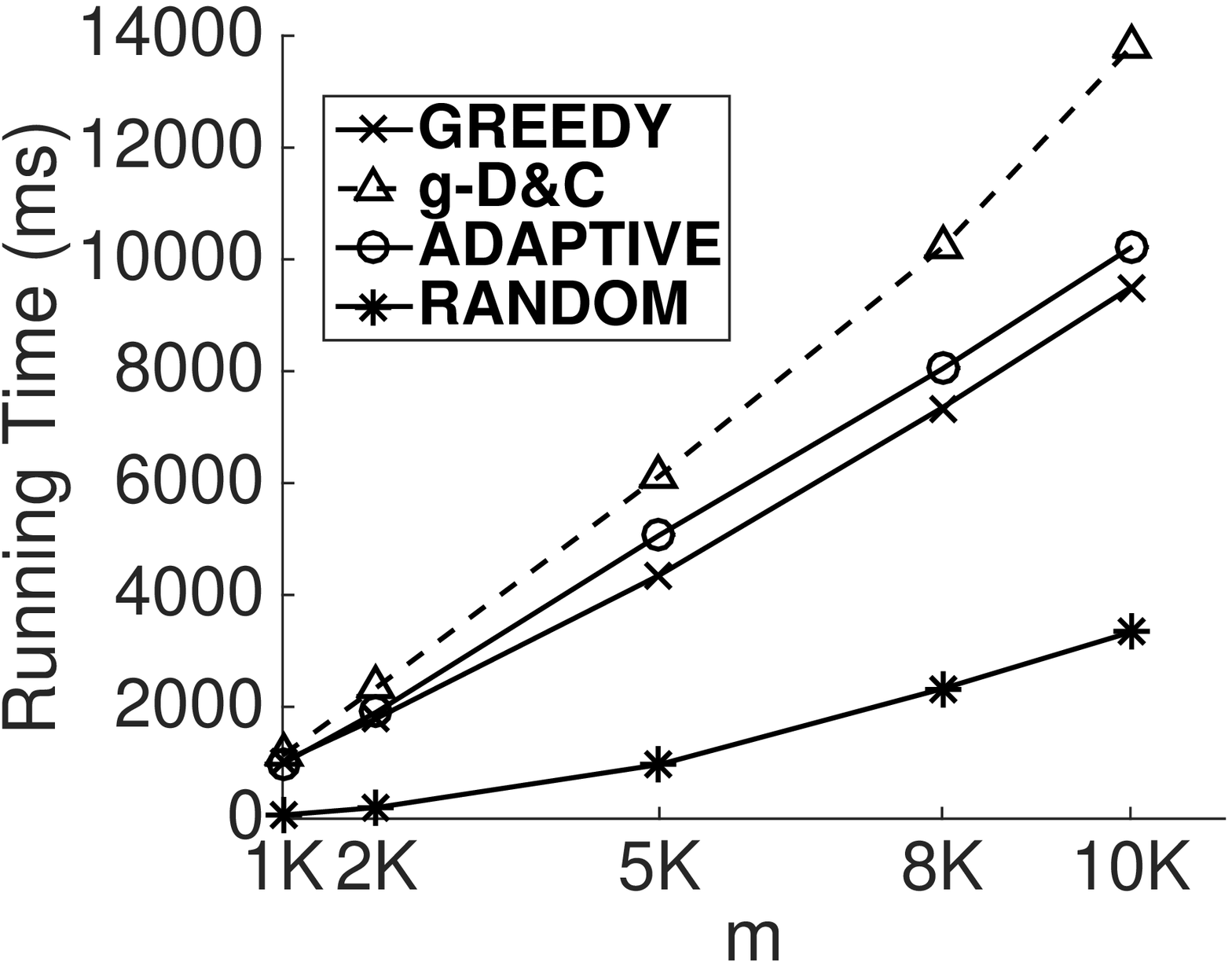}}
		\label{subfig:m_cpu}}\vspace{-2ex}
	\caption{\small Effect of the Number of Tasks $m$ (Synthetic Data).}\vspace{-4ex}
	\label{fig:tasks_m}
\end{figure}

\noindent\textbf{Effect of the Number of Workers $n$.} Figure
\ref{fig:workers_n} shows the experimental results with different
numbers of workers, $n$, from $1K$ to $10K$ over synthetic data,
where other parameters are set to their default values. Similar to
previous results about the effect of $m$,  in Figure \ref{subfig:n_score}, our proposed three
approaches can obtain good results with high assignment scores,
compared with RANDOM. Moreover, when the number, $n$, of workers
increases, the scores of all our approaches also increase. The
reason is that, when $n$ increases, we have more potential workers,
who can be assigned to nearby tasks, which may lead to even larger
scores.

In Figure \ref{subfig:n_cpu}, the running time of our approaches
increases, with the increase of the number of workers . This is due
to higher cost to process more workers (i.e., larger $n$).
Similarly,  ADAPTIVE has higher time cost than GREEDY, and lower
time cost than $g$-D\&C.

\begin{figure}[ht]\vspace{-2ex}
	\centering
	\subfigure[][{\scriptsize Scores of Assignment}]{
		\scalebox{0.2}[0.2]{\includegraphics{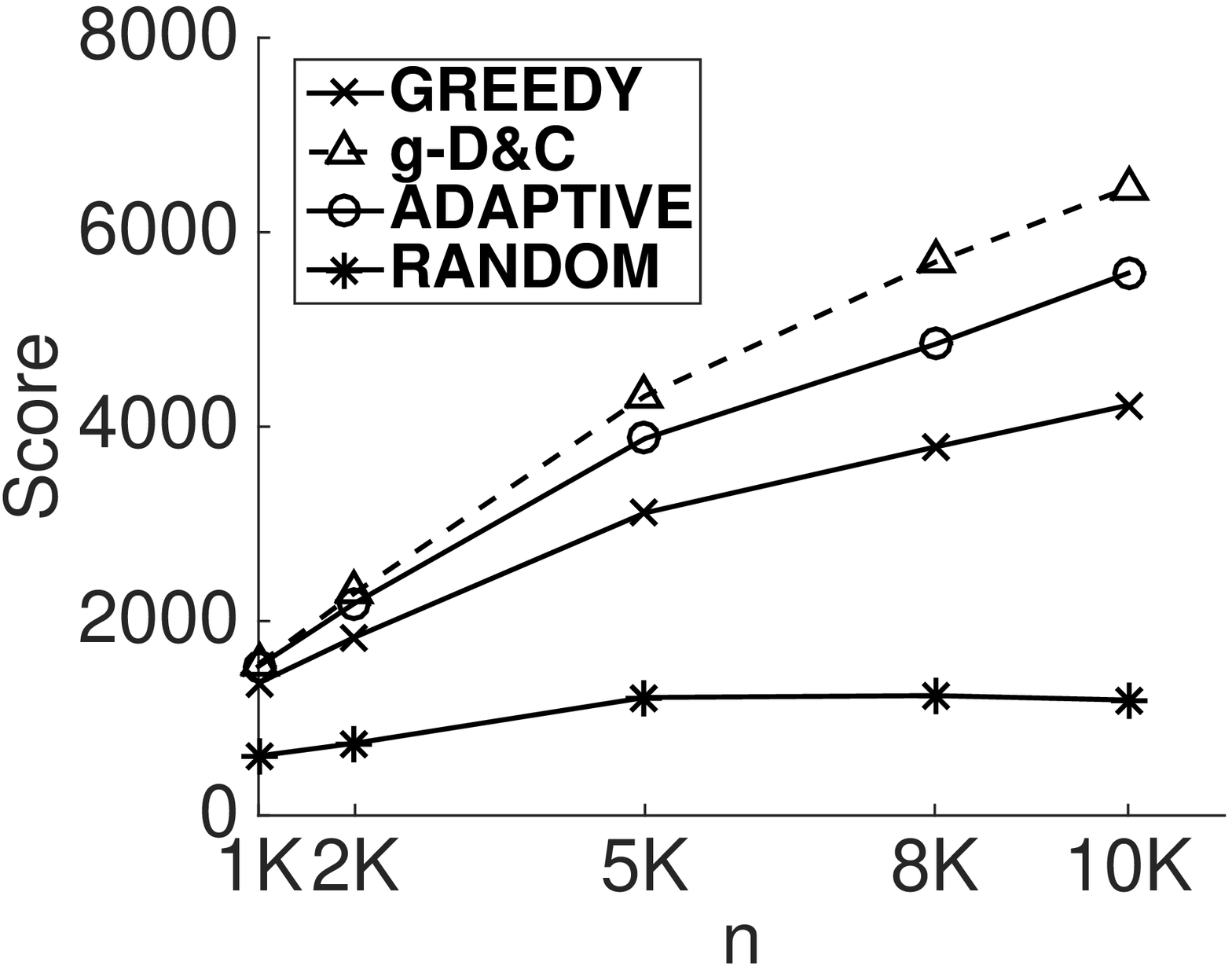}}
		\label{subfig:n_score}}
	\subfigure[][{\scriptsize Running Times}]{
		\scalebox{0.2}[0.2]{\includegraphics{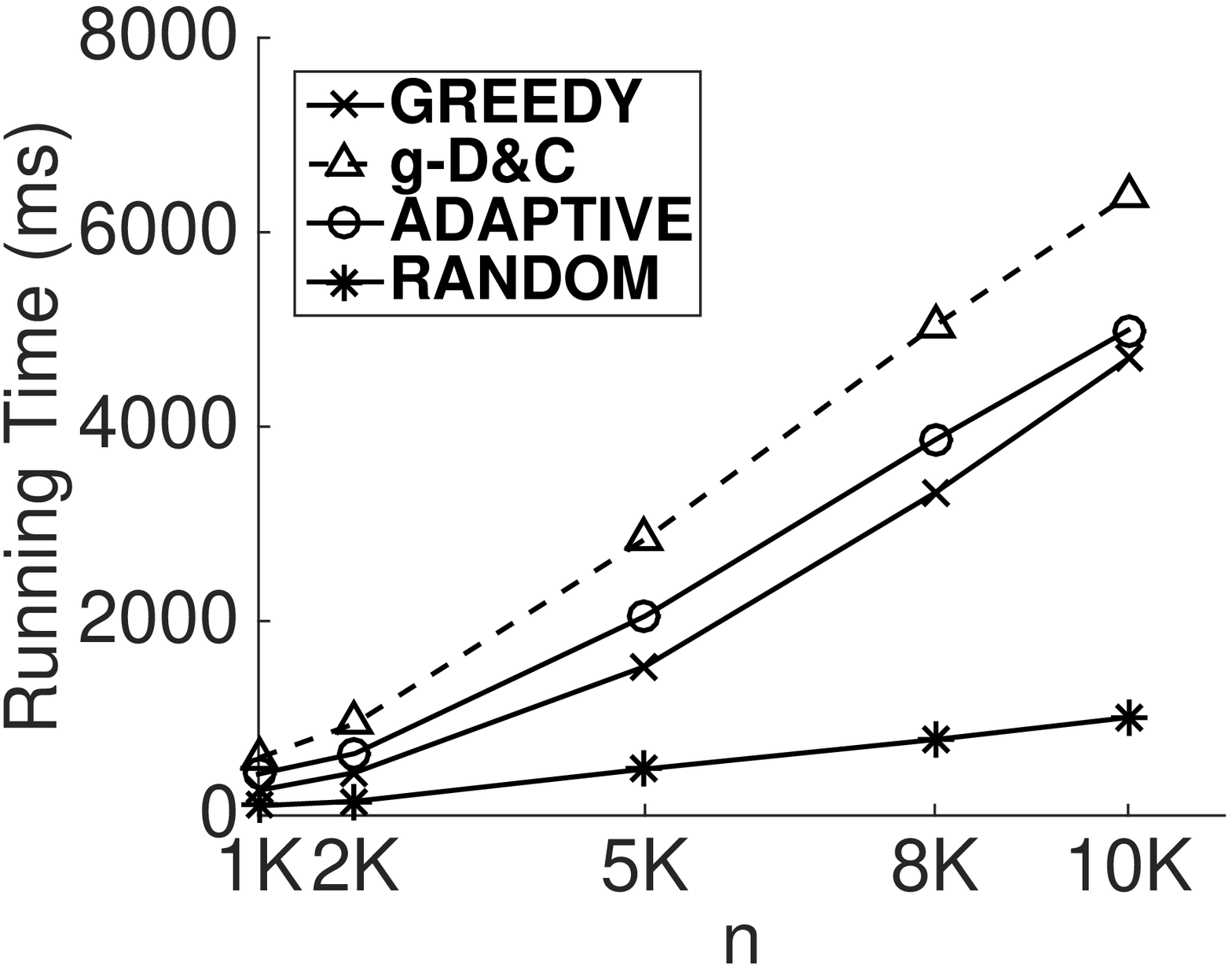}}
		\label{subfig:n_cpu}}\vspace{-2ex}
	\caption{\small Effect of the Number of Workers $n$ (Synthetic Data).}\vspace{-1ex}
	\label{fig:workers_n}\vspace{-3ex}
\end{figure}

In summary, over synthetic data sets, our ADAPTIVE algorithm trades
the accuracy for efficiency, and thus has the trade-off of
scores/times between GREEDY and $g$-D\&C.

\section{Related Work}
\label{sec:related}

Recently, with the popularity of GPS-equipped smart devices and
wireless networks (e.g., Wi-Fi and 4G), the spatial crowdsourcing
\cite{deng2013maximizing, kazemi2012geocrowd} that performs
location-based tasks has emerged and become increasingly important
in both academia and industry. In this section, we review the
related work on spatial crowdsourcing, as well as the set cover
problem (and its variants).

\vspace{0.5ex}\noindent {\bf Spatial Crowdsourcing.} Prior works on
crowdsourcing \cite{alt2010location, bulut2011crowdsourcing} usually
studied crowdsourcing problems, which  treat the location
information as a parameter and distribute tasks to workers. In these
problems, workers are not required to accomplish tasks on sites.

In contrast, the spatial crowdsourcing platform
\cite{kazemi2012geocrowd} requires workers to physically move to
some specific locations of tasks, and perform the requested
services, such as taking photos/videos, waiting in line at shopping
malls, and decorating a room. As an example, some previous works
\cite{cornelius2008anonysense, kanhere2011participatory} studied the
small-scale or specified campaigns benefiting from
\textit{participatory sensing} techniques, which utilize smart
devices (equipped by workers) to sense/collect data for real
applications.

Kazemi and Shahabi \cite{kazemi2012geocrowd} classified the spatial
crowdsourcing systems from two perspectives: people's motivation and
publishing models. From the perspective of people's motivation, the
spatial crowdsourcing can be categorized into two groups:
reward-based, in which workers can receive rewards according to the
services they supplied, and self-incentivised, in which workers
conduct tasks voluntarily. In our work, we study our MS-SC problem
based on the reward-based model, where workers are paid for doing
tasks. However, with a different goal, our MS-SC problem targets at
assigning workers to tasks by using our proposed algorithms, such
that the required skills of tasks can be covered, and the total
reward budget (i.e., flexible budget $B_j'$ in
Eq.~(\ref{eq:flexible_budget})) can be maximized. Note that, we can
embed incentive mechanisms from existing works \cite{rula2014no,
	yang2012crowdsourcing} into our MS-SC framework to distribute
rewards (flexible budgets) among workers, which is however not the
focus of our problem.

According to the publishing modes of spatial tasks, the spatial
crowdsourcing can be also classified into two categories:
\textit{worker selected tasks} (WST) and \textit{server assigned
	tasks} (SAT) \cite{kazemi2012geocrowd}. In particular, for the WST
mode, spatial tasks are broadcast to all workers, and workers can
select any tasks by themselves. In contrast, for the SAT mode, the
spatial crowdsourcing server will directly assign tasks to workers,
based on location information of tasks/workers.

Some prior works \cite{alt2010location, deng2013maximizing} on the
WST mode allowed workers to select available tasks, based on their
personal preferences. However, for the SAT mode, previous works
focused on assigning available workers to tasks in the system, such
that the number of assigned tasks on the server side
\cite{kazemi2012geocrowd}, the number of worker's self-selected
tasks on the client side \cite{deng2013maximizing}, or the
reliability-and-diversity score of assignments
\cite{cheng2014reliable} is maximized. For example, Peng et al.
\cite{cheng2014reliable} aims to obtain a worker-and-task assignment
strategy such that the assignment score (w.r.t. spatial/temporal
diversity and reliability of tasks) is maximized.

In contrast, our MS-SC problem has a different, yet more general,
goal, which maximizes the total assignment score (i.e., flexible
budget, given by the total budget of the completed tasks minus the
total traveling cost of workers). Most importantly, in our MS-SC
problem, we need to consider several constraints, such as
skill-covering, budget, time, and distance constraints. That is, the
required skill sets of spatial tasks should be fully covered by
skills of those assigned workers, which is NP-hard and intractable.
Thus, previous techniques
\cite{cheng2014reliable,deng2013maximizing,kazemi2012geocrowd} on
different spatial crowdsourcing problems cannot be directly applied
to our MS-SC problem.

Some research communities studied the theory of SAT problems and
developed some SAT (Satisfiability) solvers. However, standard SAT
solvers can only solve decision problems (i.e., NP-complete
problems), but not optimization problems (i.e., NP-hard problems,
like our MS-SC problem). Thus, we need to design specific heuristic
algorithms for tackling the MS-SC problem.

Moreover, some previous works \cite{pournajaf2014spatial,
	to2014framework} utilized \textit{differential privacy} techniques
\cite{dwork2008differential} to protect the location information,
which is used to do the assignment, but may release some sensitive
location/trajectory data (leading to malicious attacks).
Nevertheless, this privacy issue is out of the scope of this paper.

\vspace{0.5ex}\noindent {\bf Set Cover Problem.} As mentioned in
Lemma \ref{lemma:np}, the \textit{set cover problem} (SCP) is a
classical NP-hard problem, which targets at choosing a set of
subsets to cover a universe set, such that the number of the
selected subsets is minimized. SCP is actually a special case of our
MS-SC problem, in which there exists only one spatial task. However,
in most situations, we have more than one spatial task in the
spatial crowdsourcing system, which is more complex, and thus more
challenging, to tackle.

A direct variant of SCP is the \textit{weighted set cover problem},
which associates each subset with a weight. The well-known greedy
algorithm \cite{vazirani2013approximation} can achieve an
approximation ratio of $\ln(N) ( \approx H(N)$ here $H(N) =
\sum_{i=1}^{N}1/i)$, where $N$ is the size of the universe set.
Other SCP variants, such as the \textit{set multicover problem}
(SMC) and \textit{multiset multicover problem} (MSMC), focused on
covering each element of the universe set for at least specified
times using sets (in SMC, each element in subsets has just one copy)
or multisets (in MSMC, each element in subsets has a specified
number of copies). Sun and Li \cite{sun2005mechanism} studied
\textit{set cover games problem}, which covers multiple sets. However, they focused on designing a good mechanism to enable each single task to obtain a local optimal result. In contrast, our work aims to obtain a global optimal solution to maximize the score of assignment.

Different from SCP and its variants that cover only one universe
set, our MS-SC problem is targeting to cover multiple sets, such
that the assignment score is maximized. Furthermore, our MS-SC
problem is also constrained by budget, time, and distance, which is
much more challenging than SCP. To the best of our knowledge, no
prior works on SCP (and its variants) have studied the MS-SC
problem, and existing techniques cannot be used directly to tackle
the MS-SC problem.

\section{Conclusion}
\label{sec:conclusion}

In this paper, we propose the problem of the \textit{multi-skill
	oriented spatial crowdsourcing} (MS-SC), which assigns the
time-constrained and multi-skill-required spatial tasks with
dynamically moving workers, such that the required skills of tasks
can be covered by skills of workers and the assignment score is
maximized. We prove that the processing of the MS-SC problem is
NP-hard, and thus we propose three approximation approaches (i.e.,
greedy, $g$-D\&C, and cost-model-based adaptive algorithms), which
can efficiently retrieve MS-SC answers. Extensive experiments have
shown the efficiency and effectiveness of our proposed MS-SC
approaches on both real and synthetic data sets.

\balance

\bibliographystyle{abbrv}
\bibliography{add_short}

\appendix
\subsection{Proof of Lemma \ref{lemma:np}}
\begin{proof}
	We prove the lemma by a reduction from the set cover problem (SCP).
	A set cover problem can be described as follows: Given a universe
	set $U = \{a_1, a_2, ..., a_n\}$ and m subsets $Z_1, Z_2, ...,$ $Z_m
	\subseteq U$. For each subset $Z_i$, it is associated with a cost
	$c_i$. The set cover problem is to find a set $K \subseteq \{1, 2,
	..., m\}$ that minimizes $\sum_{i \in K}c_i$, such that $\cup_{i \in
		K}Z_i = U$.
	
	For a given set cover problem, we can transform it to an instance of
	MS-SC as follows: at timestamp $p$, we give only one task $t_j$ with
	required skills set $Y_j \subseteq \Psi$ and $Y_j = U$, whose budget
	value, $B_j = \sum_{w_i \in W}c_{ij}$, is big enough to hire workers
	to satisfy the required skills. In addition, the arrival deadline of
	the task is late enough for any worker to arrive in time. For $m$
	workers, each worker $w_i$ has a skills set $X_i$, such that $X_i =
	Z_i$, and a cost $c_{ij} = c_i$. In addition, the maximum moving
	distance $d_i$ is larger than $dist(l_i(p), l_j)$. Then, for this
	MS-SC instance, we want to select $K' \subseteq \{1, 2, ..., m\}$
	workers to support task $t_j$ that maximizes the score, $S_p = B_j -
	\sum_{i \in K'}c_{ij}$ and $Y_j \subseteq \cup_{i \in K'}X_i$.
	
	The answer of this MS-SC is:
	\begin{flalign}
	&\text{maximize } \left(B_j - \sum_{i \in I'}c_{ij}\right) \notag\\
	\implies &\text{minimize } \sum_{i \in K'}c_{ij}\notag
	\implies \text{minimize } \sum_{i \in K}c_{i}\notag
	\end{flalign}

	As $B_j$ is a constant, to maximize the score $S_p$ is same as to
	minimize the total cost, $\sum_{i \in K'}c_{ij}$, of the assigned
	workers, which is identical to $\sum_{i \in K}c_{i}$. Given this
	mapping it is easy to show that the set cover problem instance can
	be solved if and only if the transformed MS-SC problem can be
	solved.
	
	This way, we can reduce SCP to the MS-SC problem. Since SCP is known
	to be NP-hard \cite{vazirani2013approximation}, MS-SC is also
	NP-hard, which completes our proof.
\end{proof}

\subsection{Proof of Lemma \ref{lemma:dominate_worker}}
\begin{proof}
	We want to prove that, if worker $w_a$ is assigned to task $t_j$,
	then any skill $k_a \in X_a \cap Y_j$ can be also covered by worker
	$w_b$. From the lemma assumption, since worker $w_a$ is dominated by
	$w_b$, we have $X_a \subseteq X_b$ and $c_{aj}\geq c_{bj}$. Thus,
	from the condition that $X_a \subseteq X_b$, we have $\frac{|X_a
		\cap (Y_j - \widetilde{Y_j})|}{|Y_j|}\leq \frac{|X_b \cap (Y_j -
		\widetilde{Y_j})|}{|Y_j|}$. Moreover, since $c_{aj}\geq c_{bj}$,
	from Eq.~(\ref{eq:score_increase}), we have $\Delta S_p(w_a)\leq
	\Delta S_p(w_b)$. Therefore, worker $w_a$ is not better than worker
	$w_b$, in terms of the score increase. Hence, we can safely prune
	the worker-and-task pair $\langle w_a, t_j\rangle$.
\end{proof}

\subsection{Proof of Lemma \ref{lemma:expensive_worker}}
\begin{proof}
	From Definition \ref{definition:MS_SC}, we have the budget
	constraint that $\sum_{\forall \langle w_i, t_j\rangle \in
		I_p}c_{ij}$ $\leq B_j$. From the lemma assumption, if it holds that
	$c_{ij}> B_j-\widetilde{c_{\cdot j}}$, then we have $c_{ij} +
	\widetilde{c_{\cdot j}})> B_j$, which violates the constraint that
	the total traveling cost should not exceed the maximum budget $B_j$.
	Thus, we should not assign worker $w_i$ to task $t_j$.
	
	Due to the non-increasing property of the remaining budget
	$(B_j-\widetilde{c_{\cdot j}})$, for the rest of assignment rounds,
	we still cannot assign worker $w_i$ to $t_j$ (since the task cannot
	afford the traveling cost of the worker $w_i$). Hence, we can safely
	prune worker $w_i$.
\end{proof}

\subsection{Proof of Lemma \ref{lemma:prune_insufficient_task}}

\begin{proof} Since the traveling cost, $c_{ij}$, of worker $w_i$ is
	greater than the remaining budget, $(B_j-\widetilde{c_{\cdot j}})$,
	according to Lemma \ref{lemma:expensive_worker}, worker $w_i$ should
	not be assigned to task $t_j$.
	
	Therefore, we only need to prove that, for any set of the remaining
	unassigned workers $w_r \in (W(t_j)-\widetilde{W(t_j)})$ who can
	cover the required skill set $Y_j$, their total traveling cost is
	always greater than the remaining budget $(B_j-\widetilde{c_{\cdot
			j}})$.
	
	Without loss of generality, assume that we have a subset, $R$, of
	unassigned workers, $w_r$, in $(W(t_j)-\widetilde{W(t_j)})$ that can
	be assigned to task $t_j$, and cover the skill set $Y_j$ (note: if
	such a subset does not exist, then task $t_j$ cannot be fully
	covered by workers' skills and can be safely pruned). Then, we have
	the relationship of skill sets between worker $w_i$ (with the
	highest $\frac{\Delta S_p}{|X_i \cap (Y_j - \widetilde{Y_j})|}$
	value) and workers $w_r$ below:
	
	$(X_i \cap (Y_j - \widetilde{Y_j})) \subseteq \cup_{\forall w_r\in
		R} (X_r \cap (Y_j - \widetilde{Y_j})).$
	
	Alternative, we can derive the relationship of their set sizes, that
	is:

	\begin{eqnarray}
	|X_i \cap (Y_j - \widetilde{Y_j})| &\leq& |\cup_{\forall w_r\in R}
	(X_r \cap (Y_j -
	\widetilde{Y_j}))|\notag\\
	&\leq& \sum_{\forall w_r\in R}|X_r \cap (Y_j -
	\widetilde{Y_j})|.\label{eq:set_size}
	\end{eqnarray}
	
	On the other hand, according to Eq.~(\ref{eq:score_increase}) and
	our lemma assumption (i.e., worker $w_i$ has the largest value of
	$\frac{\Delta S_p}{|X_i \cap (Y_j - \widetilde{Y_j})|}$), for any
	worker $w_r\in R$, we have the following relationship between $w_i$
	and $w_r$:
	
	\begin{eqnarray}
	&&\frac{\Delta S_p(w_i, t_j)}{|X_i \cap (Y_j - \widetilde{Y_j})|}
	\geq \frac{\Delta
		S_p (w_r, t_j)}{|X_r \cap (Y_j - \widetilde{Y_j})|}\notag\\
	\Leftrightarrow && \frac{B_j}{|Y_j|} - \frac{c_{ij}}{|X_i \cap (Y_j
		- \widetilde{Y_j})|}
	\geq \dfrac{B_j}{|Y_j|} - \frac{c_{rj}}{|X_r \cap (Y_j - \widetilde{Y_j})|} \notag\\
	\Leftrightarrow && c_{rj} \geq c_{ij}\cdot \frac{|X_r \cap (Y_j -
		\widetilde{Y_j})|}{|X_i \cap (Y_j - \widetilde{Y_j})|}\notag
	\end{eqnarray}
	
	As a result, the total traveling cost for all $w_r \in R$ has the
	property below:
	\begin{eqnarray}
	\sum_{\forall w_r\in R} c_{rj} \geq c_{ij}\cdot \frac{\sum_{\forall
			w_r\in R} |X_r \cap (Y_j - \widetilde{Y_j})|}{|X_i \cap (Y_j -
		\widetilde{Y_j})|}. \label{eq:pruning_derivation}
	\end{eqnarray}
	
	By combining Eq.~(\ref{eq:set_size}) with
	Eq.~(\ref{eq:pruning_derivation}), we can apply the inequality
	transition, and obtain:
	
	\begin{eqnarray}
	\sum_{\forall w_r\in R} c_{rj} \geq c_{ij}.
	\end{eqnarray}

	Since it holds that $c_{ij}> B_j-\widetilde{c_{\cdot j}}$ by the
	lemma assumption, we thus have:
	\begin{eqnarray}
	\sum_{\forall w_r\in R} c_{rj} >B_j-\widetilde{c_{\cdot j}},
	\end{eqnarray}
	\noindent which exactly indicates that any subset, $R$, of those
	unassigned workers has the total traveling cost exceeding the
	remaining budget. Hence, we can safely prune task $t_j$, and the
	lemma holds.
\end{proof}

\subsection{MS-SC Grid Index}

In order to facilitate the processing of the MS-SC problem, we present an efficient cost-model-based indexing
mechanism, which can maintain workers and tasks and help the
retrieval of MS-SC answers.

\noindent {\bf Index Structure.}
We first introduce the index structure, namely {\sf MS-SC-Grid}, for
the MS-SC system. In particular, we divide a 2-dimensional data
space, $[0,1]^2$, into $1/\tau^2$ square cells with side length
$\tau$, where $\tau < 1$. Similar to the cost model of the grid
index in \cite{cheng2014reliable}, by utilizing the power
law \cite{belussi1998self} for the \textit{correlation fractal dimension} $D_2$ of
the tasks and workers in the 2D data space, we can construct a cost
model of updating the grid index after insert or delete a worker with respect to the side length
$\tau$. Then we can estimate the best value for $\tau$ to minimize the update cost.

Below, we will illustrate the content of each cell, $cell_i$, in the
grid index, which includes worker/task lists, statistics of
workers/tasks, and a cell list.

\noindent {\bf Worker and Task Lists in Cells.} Within the grid
index, each cell, $cell_i$, has a unique ID, $cid$, and is
associated with two lists, which store sets of tasks and workers,
respectively, that reside in the cell. In the worker list, we
maintain records, in the form of sextuple:
$$\langle wid, l, C, d, v, X\rangle,$$
\noindent where $wid$ is the worker ID, $l$, $C$, $d$, and $v$
represent the location, constant related to the traveling cost of
the distance, maximum moving distance, and velocity of the worker,
respectively, and $X$ is a set of skills that the worker has.

In the task list, we keep records in the form:
$$\langle tid, l, e, B, Y\rangle,$$
\noindent where $tid$ is the task ID, $l$ is the location of the
task, $e$ denotes the arrival deadline of the task, $B$ represents
the budget of the task, and $Y$ is the set of skills required by the
task.

\noindent {\bf Statistics/Aggregates in Cells.} For each cell
$cell_i$, we maintain statistics/aggregates for workers and tasks in
it, including:
\begin{itemize}
	\item the minimum constant for the traveling cost $C^{(i)}_{min}$;
	\item the largest maximum moving distance $d^{(i)}_{max}$;
	\item the maximum velocity, $v^{(i)}_{max}$, for all workers in the cell;
	\item the latest arrival deadline $e^{(i)}_{max}$;
	\item the maximum budget, $B^{(i)}_{max}$, for all tasks in the cell;
	\item the union, $X^{(i)}_{cell}$, of sets of workers' skills; and
	\item the union, $Y^{(i)}_{cell}$, of sets of the required skills by tasks in the cell.
\end{itemize}

\underline{\it Synopses for Skill Sets.} In order to time- and
space-efficiently organize/manipulate the sets of skills for workers
and tasks (i.e., $X^{(i)}_{cell}$ and $Y^{(i)}_{cell}$ above,
respectively), an alternative is to maintain two bitmap synopses,
$BM_{X}$ and $BM_{Y}$, in which each bit corresponds to a skill.
That is, for any skill in the skill set $X^{(i)}_{cell}$ (or
$Y^{(i)}_{cell}$), its corresponding bit in bitmap $BM_{X}$ (or
$BM_{Y}$) is set to ``1''; the remaining bits in $BM_{X}$ (or
$BM_{Y}$) are set to ``0''. This way, we can apply bit operations
(e.g., bit-AND or bit-OR) between any two synopses, and check the
relationship (e.g., containment or intersection) between their
corresponding skill sets.

\noindent {\bf Cell List.} Furthermore, each cell $cell_i$ is also
associated with a cell list, $clist^{(i)}$, which contains all the
cell IDs that can be reachable to at least one worker in cell
$cell_i$.

\noindent {\bf Pruning Strategy on the Cell Level.}
When calculate the valid worker-and-task pairs for worker $w_i$ in
our approaches, we can use the grid index to accelerate the
searching time by just checking the cells in the $clist^{(j)}$ of
the cell $cell_j$  and the cell $cell_j$ itself, where worker $w_i$
locates in cell $cell_j$. However, we do not need to check all the
tasks in the cells in $clist^{(j)}$. We propose 4 pruning strategies
to further reduce the search space.

Assuming we are searching the valid worker-and-task pairs for worker
$w_i$ located in cell $cell_j$, we now show how to prune a cell
$cell_k$ in $clist^{(j)}$ before further checking all the tasks in
$cell_k$. We first calculate the minimum distance $MIND_{k}$ between
the location $l_i$ of worker $w_i$ and any points in $cell_k$. Then
we have 4 pruning strategies.

\begin{itemize}
	\item If $d_i < MIND_{k}$, all the tasks in $cell_k$ are out of the working range of worker $w_i$, that is, worker $w_i$ will not accept any task in $cell_k$;
	\item If $MIND_{k}/v_i > e^{(i)}_{max}$, worker cannot arrive any
	task in $cell_k$ before the arrival deadline of that task;
	\item If $MIND_{k}C_i > B^{(k)}_{max}$, no task in $cell_k$ can
	afford the traveling cost of worker $w_i$;
	\item If $X_i \cap Y^{(k)}_{cell} = \emptyset$, worker $w_i$ cannot support any skills of any task in $cell_k$.
\end{itemize}

If any one strategy is met, we can safely prune cell $cell_k$. After pruning those unreachable or
unsupportable cells, we further check the rest cells one by
one to construct the valid worker-and-task pairs for worker $w_i$.

\noindent {\bf Dynamic Index Maintenance.} Since workers/tasks can
join and leave the spatial crowdsourcing system, our grid index
should be efficient for handling worker/task updates. Specifically,
for each incoming/expired worker/task, we update contents of cells
that workers/tasks fall into, as well as the cell list. Due to space
limitations, we will not discuss insertion/detetion of workers/tasks
in detail.

\begin{figure}[ht]\vspace{-2ex}\centering
	\scalebox{0.3}[0.3]{\includegraphics{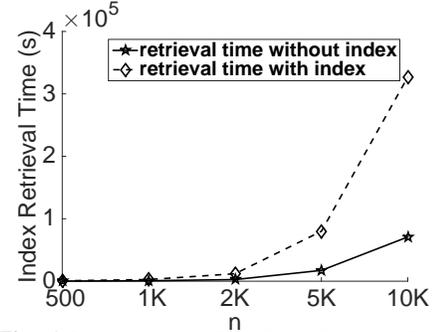}}\vspace{-2ex}
	\caption{\small Worker-and-Task Pairs Retieval Time}
	\label{fig:index}\vspace{-3ex}
\end{figure}
Figure \ref{fig:index} presents the \textit{index retrieval time} 
(i.e., the time cost of retrieving valid worker-and-task pairs) 
over UNIFORM data, where $n=5m$ and $m$ varies from 500 to 
10K. The MS-SC grid index can reduce the time of obtaining 
worker-and-task pairs dramatically (up to 79\%), compared 
	with the running time of directly enumerating and checking 
	all possible worker-and-task pairs.

\subsection{Effects of Moving Distance and Expiration Times}
In this subsection, we show the effects of the range of maximum moving distances of workers, $d$, and the range of the expiration time of tasks, $rt$.

\noindent\textbf{Effect of Range of Workers' Maximum Moving Distances $d$.}
Figure \ref{fig:distance_d} shows the effect of the range 
$[d^-, d^+]$ of worker's maximum moving distances on 
the scores of assignments and the running times, where 
we vary the range of $d$ from $[0.1, 0.2]$ to $[0.4, 0.5]$.
In Figure \ref{subfig:d_score}, all the 3 approaches can 
achieve good scores of assignment. They still has a 
comparison sequence on the score of the results: the 
score of the results obtained by ADAPTIVE is highest 
among other approaches in our experiments.
Then, $g$-D\&C can also get a higher score of assignment 
than GREEDY. Similar to the discuss of the effect of the 
velocities of workers, the increase of $d$ enlarges the access 
range of workers at the beginning. However, when the
constraint of $d$ is relaxed, the constraints from other 
parameters prevent the scores keeping growing. 

For the running times, Adaptive is faster than $g$-D\&C, but slower than Greedy. Random runs fastest, however, least effectively.

\begin{figure}[ht]\vspace{-2ex}
	\centering
	\subfigure[][{\scriptsize Scores of Assignment}]{
		\scalebox{0.2}[0.2]{\includegraphics{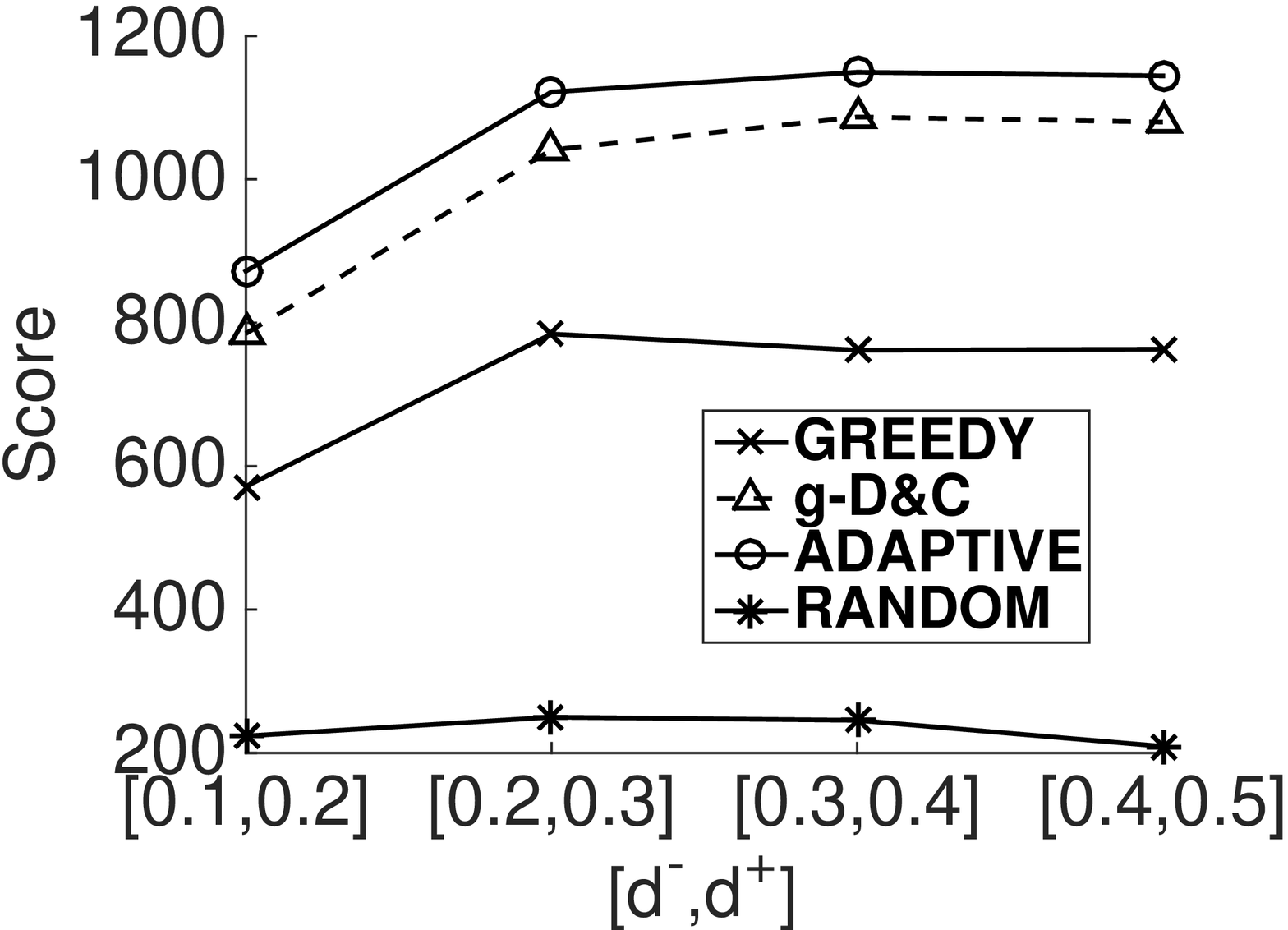}}
		\label{subfig:d_score}}
	\subfigure[][{\scriptsize Running Times}]{
		\scalebox{0.2}[0.2]{\includegraphics{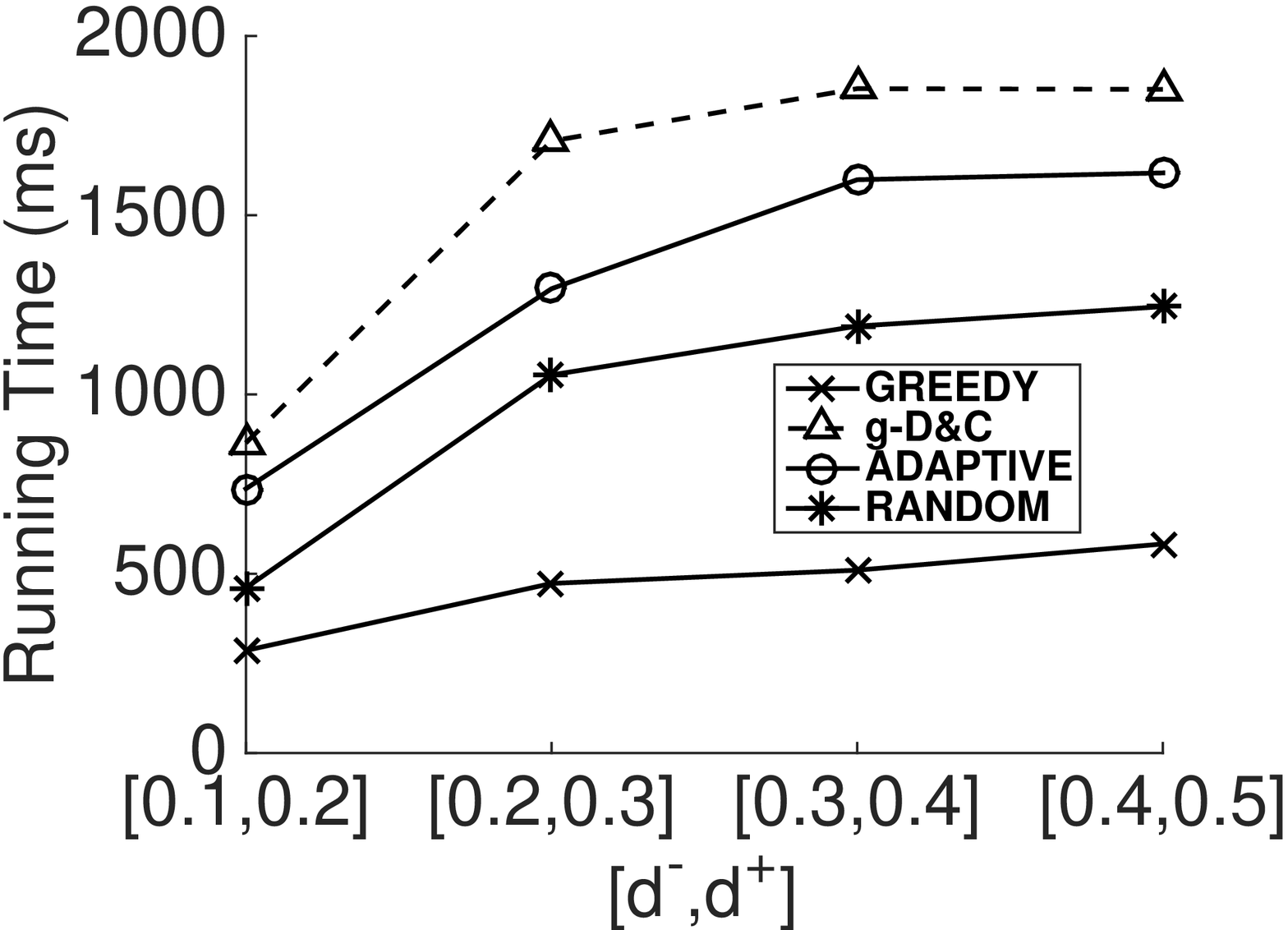}}
		\label{subfig:d_cpu}}\vspace{-2ex}
	\caption{\small Effect of Range of Maximum Moving Distances $[d^-, d^+]$ (Real Data).}\vspace{-4ex}
	\label{fig:distance_d}
\end{figure}

\noindent\textbf{Effect of Range of Expiration Times of Tasks $rt$.}
Figure \ref{fig:expiration_e} shows the effect of the range 
$[rt^-, rt^+]$ of tasks' expiration times on 
the scores of assignments and the running times, where 
we vary the range of $rt$ from $[0.25, 0.5]$ to $[3, 4]$.
In Figure \ref{subfig:e_score}, all the 3 approaches can 
achieve good scores of assignment. ADAPTIVE still obtains highest scores compared with other approaches, and $g$-D\&C is better than Greedy in scores, however not as good as Adaptive. Random just can receive lower scores.
Similar to the discuss of the effect of the 
velocities of workers, the increase of $rt$ enlarges the access 
range of workers at the beginning. However, when the
constraint of $rt$ is relaxed, the constraints from other 
parameters prevent the scores from keeping growing.

\begin{figure}[ht]
	\centering
	\subfigure[][{\scriptsize Scores of Assignment}]{
		\scalebox{0.2}[0.2]{\includegraphics{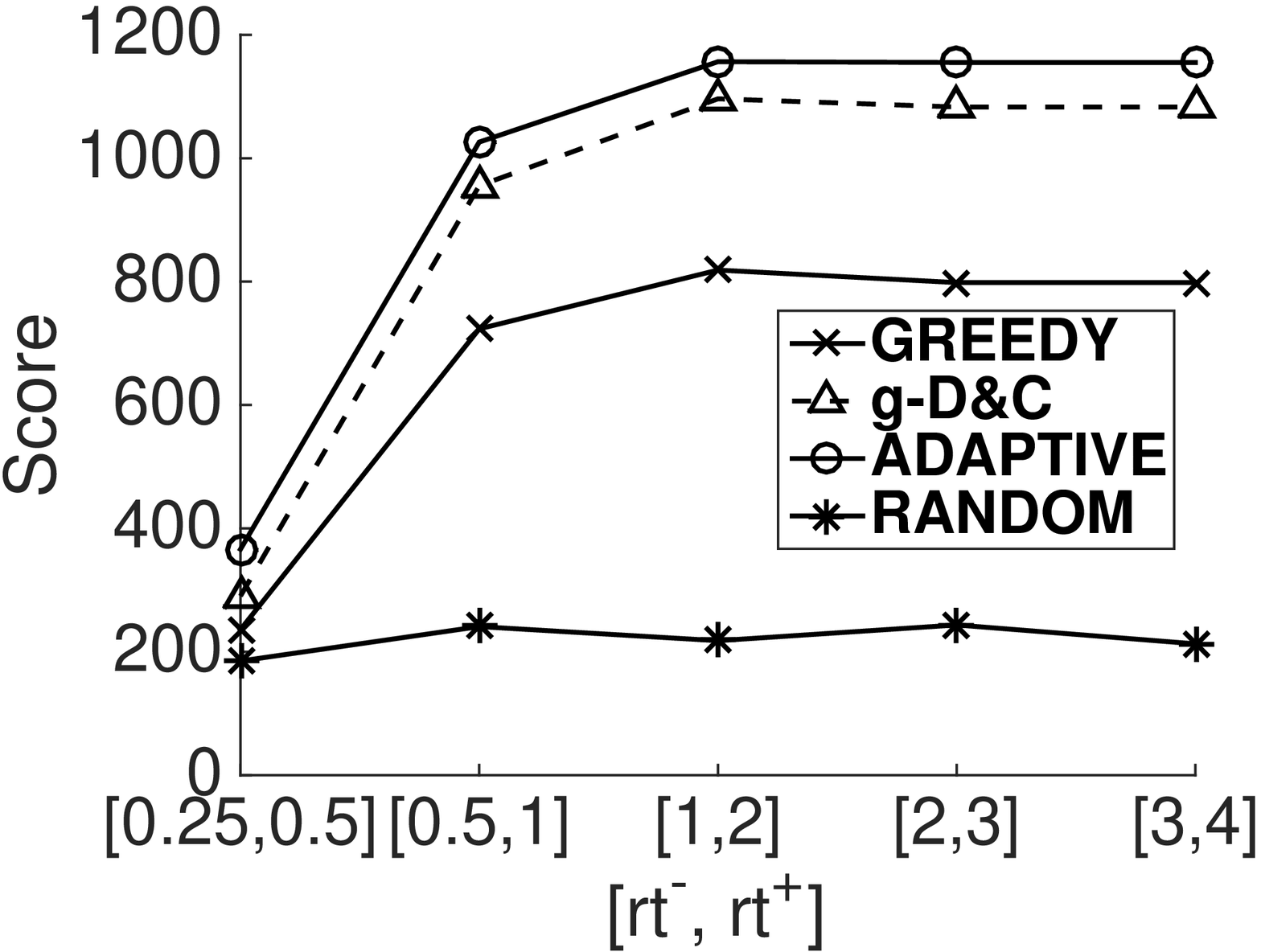}}
		\label{subfig:e_score}}
	\subfigure[][{\scriptsize Running Times}]{
		\scalebox{0.2}[0.2]{\includegraphics{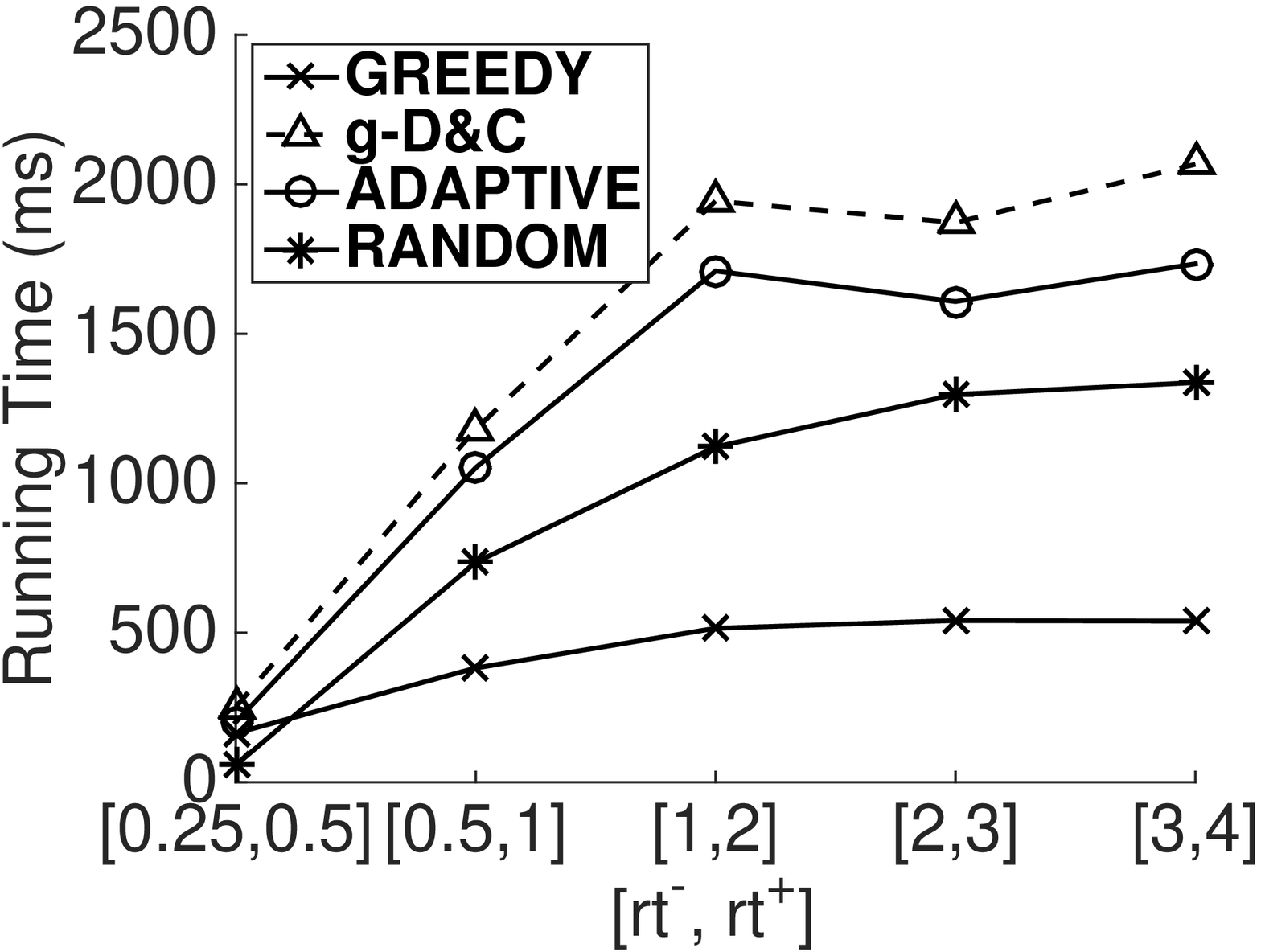}}
		\label{subfig:e_cpu}}
	\caption{\small Effect of Range of Expiration Times of Task $[rt^-, rt^+]$ (Real Data).}\vspace{-5ex}
	\label{fig:expiration_e}
\end{figure}
For the running times, all the approaches use more time when the range of expiration times increases, which is because workers can arrive at more tasks before their arrival deadlines leading to the problem space increases. Comparatively, Adaptive is faster than $g$-D\&C, but slower than Greedy. Random runs fastest, however least effectively.

\subsection{Effect of Number of Task and Workers (SKEWED)}

\noindent\textbf{Effect of Number of Tasks $m$.}
Figure \ref{fig:tasks_m_s} illustrates the effect of the 
number of tasks, when $m$ changing from 1K to 10K, 
and the distribution of the location of workers and tasks are SKEWED.
For the scores of assignments in Figure \ref{subfig:m_score_s}, 
$g$-D\&C obtains results with highest scores. ADAPTIVE 
performs similar to $g$-D\&C and also achieves good results. 
GREEDY is not as good as $g$-D\&C and ADAPTIVE, but still 
much better than RANDOM. All the approaches can 
obtain a higher score when the number of tasks become 
larger. In Figure \ref{subfig:m_cpu_s}, 
the running times increase when the number of tasks increase, because the problem space increases. The Adaptive is slower than Greedy and faster than $g$-D\&C. 

\begin{figure}[ht]\vspace{-2ex}
	\centering
	\subfigure[][{\scriptsize Scores of Assignment}]{
		\scalebox{0.2}[0.2]{\includegraphics{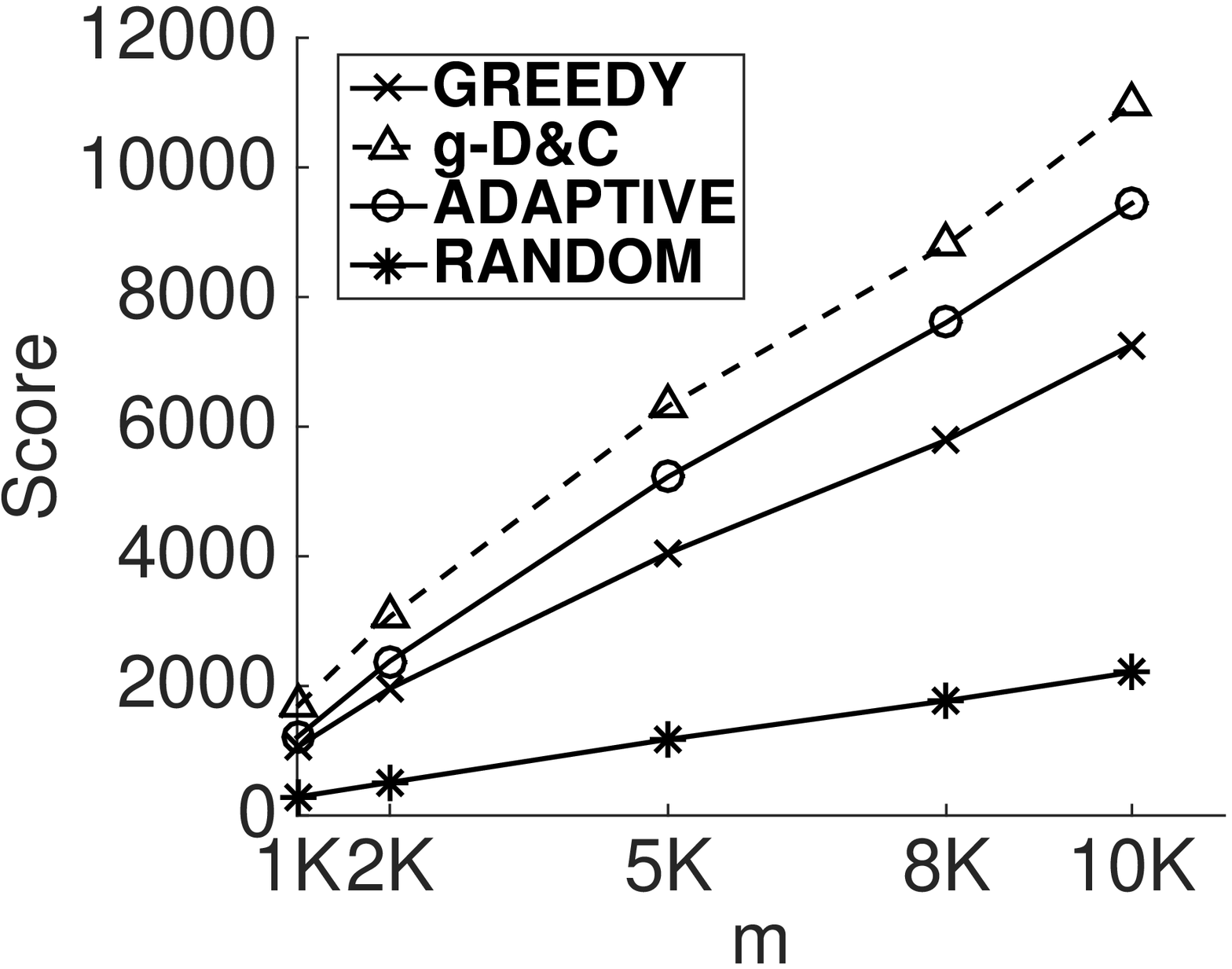}}
		\label{subfig:m_score_s}}
	\subfigure[][{\scriptsize Running Times}]{
		\scalebox{0.2}[0.2]{\includegraphics{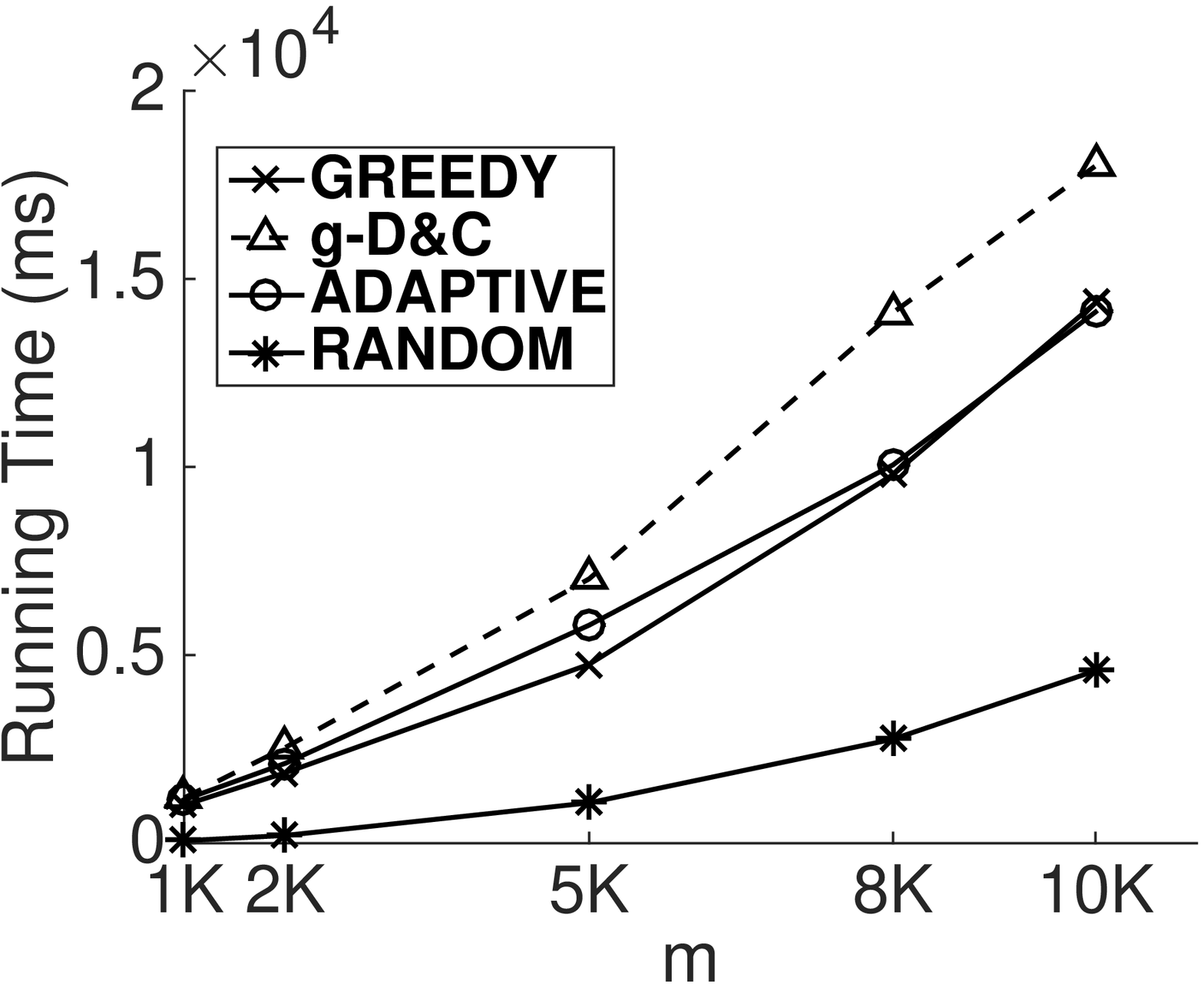}}
		\label{subfig:m_cpu_s}}\vspace{-2ex}
	\caption{\small Effect of Number of Tasks $m$ (Synthetic Data).}\vspace{-4ex}
	\label{fig:tasks_m_s}
\end{figure}

\noindent\textbf{Effect of Number of Workers $n$.}
Figure \ref{fig:workers_n_s} shows the experimental results 
when the number of workers, $n$, changing from 1K to 
10K, and the distribution of the location of workers and tasks 
are SKEWED. Similar to the results in previous discussion of the 
effect of $m$, our three approaches can obtain good results 
with high scores of assignment. In addition, the scores 
of all the approaches increase when the number of workers 
increases. The reason is that, when the number of workers increases, we 
have more workers, who are close to tasks and cost 
less, to select, which leads to the score increases. 
In Figure \ref{subfig:n_cpu_s}, the running times increase when 
the number of workers increases, because the space of problem 
increases when there are more workers. The speed of Adaptive is 
higher than $g$-D\&C and lower than Greedy.

\begin{figure}[ht]\vspace{-2ex}
	\centering
	\subfigure[][{\scriptsize Scores of Assignment}]{
		\scalebox{0.2}[0.2]{\includegraphics{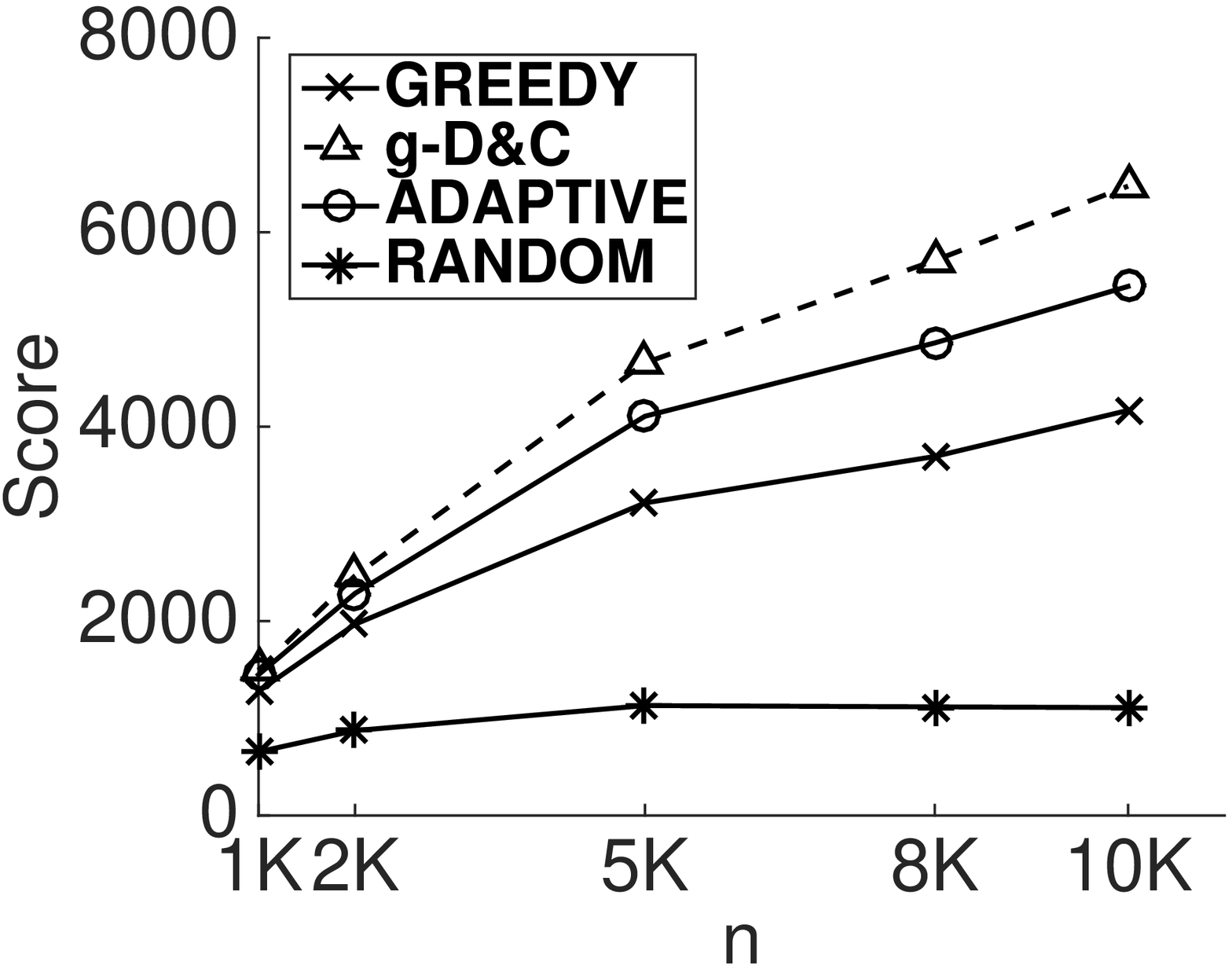}}
		\label{subfig:n_score_s}} 
	\subfigure[][{\scriptsize Running Times}]{
		\scalebox{0.2}[0.2]{\includegraphics{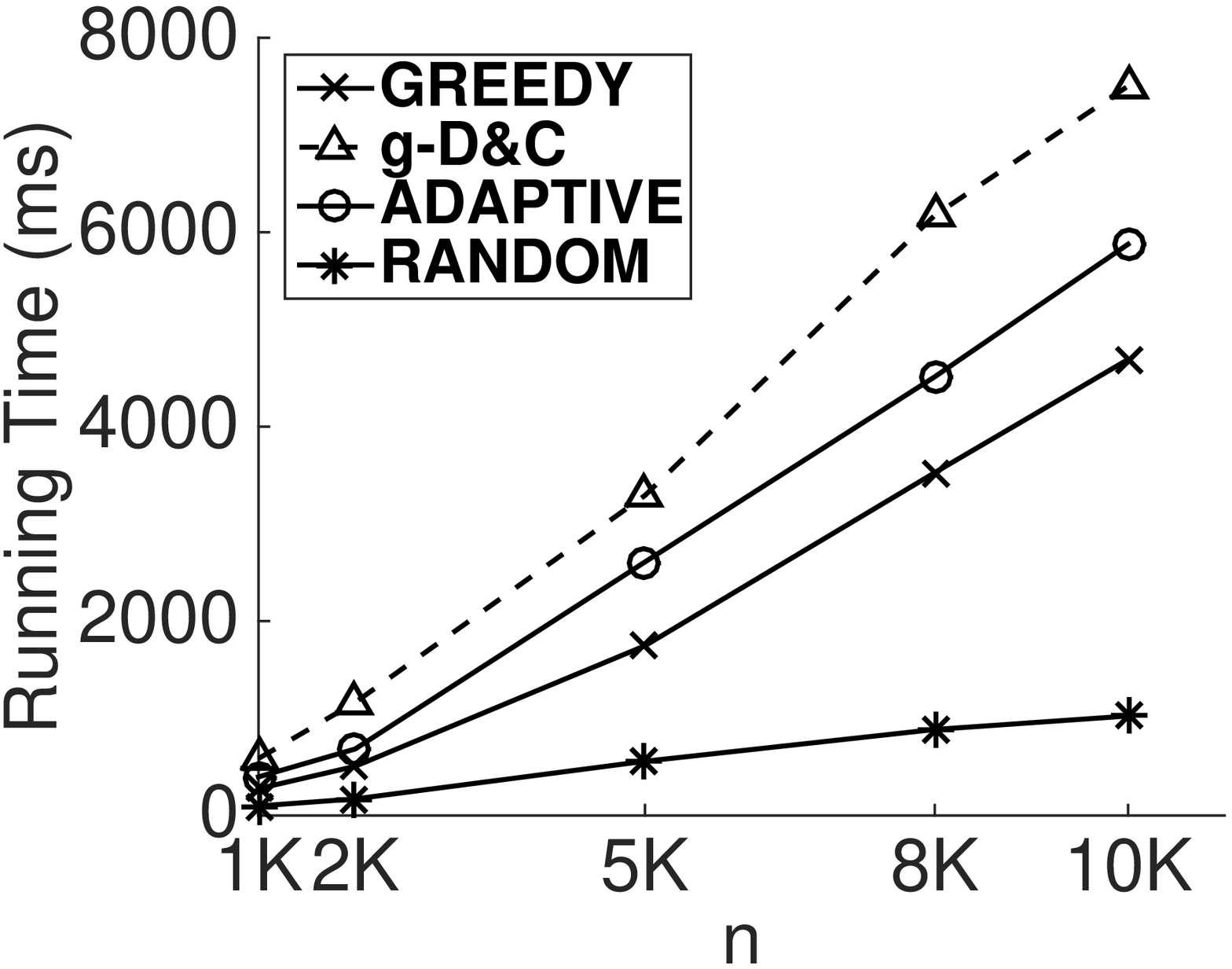}} 
		\label{subfig:n_cpu_s}}\vspace{-2ex}
	\caption{\small Effect of Number of Workers $n$ (Synthetic Data).}\vspace{-1ex}
	\label{fig:workers_n_s}\vspace{-4ex}
\end{figure}

\end{document}